%% file: main.tex
\documentclass[reprint,floatfix,superscriptaddress,nofootinbib]{revtex4-2}
\usepackage[T1]{fontenc}
\input{packages.tex}
\usepackage{float}
\usepackage{natbib}
\theoremstyle{plain}
\usepackage{thmtools, thm-restate}
\newtheorem{theorem}{Theorem}[section]

\newtheorem{lemma}[theorem]{Lemma}

\newtheorem{definition}[theorem]{Definition}

\usepackage[capitalise,noabbrev]{cleveref}
\usepackage{ragged2e}
\PassOptionsToPackage{hyphens}{url}
\usepackage{hyperref}

\usepackage{pgffor}

\begin{document}
\title{Exploiting the equivalence between quantum neural networks and perceptrons}
\author{Chris Mingard}
\affiliation{Rudolf Peierls Centre for Theoretical Physics,  University of Oxford; Oxford, OX1 3PU, UK}
\affiliation{Physical and Theoretical Chemistry Laboratory, University of Oxford; Oxford, OX1 3QZ, UK}
\author{Jessica Pointing} 
\affiliation{Department of Physics,  University of Oxford; Oxford, OX1 3PU, UK}

\author{Charles London} 
\affiliation{Department of Computer Science, University of Oxford; Oxford, OX1 3QG, UK}

\author{Yoonsoo Nam}
\affiliation{Rudolf Peierls Centre for Theoretical Physics, University of Oxford; Oxford, OX1 3PU, UK}

\author{Ard A. Louis} \thanks{email: ard.louis@physics.ox.ac.uk} 
\affiliation{Rudolf Peierls Centre for Theoretical Physics,  University of Oxford; Oxford, OX1 3PU, UK}

\date{\today}

\begin{abstract}
Quantum machine learning models based on parametrized quantum circuits, also called quantum neural networks (QNNs), are considered to be among the most promising candidates for applications on near-term quantum devices.  Here we explore the expressivity and inductive bias of QNNs by exploiting an exact mapping from QNNs with inputs $x$ to classical perceptrons acting on $x \otimes x$ (generalised to complex inputs).   The simplicity of the perceptron architecture allows us to provide clear examples of the shortcomings of current QNN
models, and the many barriers they face to becoming useful
general-purpose learning algorithms. For example, a QNN with amplitude encoding cannot express the Boolean parity function for $n\geq 3$, which is but one of an exponential number of data structures that such a QNN is unable to express. Mapping a QNN to a classical perceptron simplifies training, allowing us to systematically study the inductive biases of other, more expressive embeddings on Boolean data.  Several popular embeddings primarily produce an inductive bias towards functions with low class balance,  reducing their generalisation performance compared to deep neural network architectures which exhibit much richer inductive biases.  %Like perceptrons, current QNNs do not automatically learn features -- a key capability of deep neural networks.  
We explore two alternate strategies that move beyond standard QNNs. In the first, we use a QNN to help generate a classical DNN-inspired kernel. In the second, we draw an analogy 
 to the hierarchical structure of deep neural networks and construct a layered non-linear QNN that is provably fully expressive on Boolean data, while also exhibiting a richer inductive bias than simple QNNs. % This translates into improved generalisation performance on the Boolean dataset and on a simplified FashionMNIST dataset.  
Finally, we discuss characteristics of the  QNN literature that may obscure how hard it is to achieve quantum advantage over deep learning algorithms on classical data.
\end{abstract}

\maketitle

\section{Introduction}\label{sec:introduction}

Machine learning has undergone a revolution with the advent of large-scale deep (classical) neural network models (DNNs), which allow for the automatic extraction of features from data and the learning of complex functions. DNNs are highly expressive, satisfying universal approximation theorems \citep{cybenko1989approximation,hornik1989multilayer}, and can fit arbitrary functions on large datasets \citep{zhang2016understanding}. Despite this high capacity, these models often generalise well rather than overfitting as classical learning theory would predict. This means they have a strong inductive bias towards functions that describe real-world data \citep{arpit2017closer,valle2020generalization,canatar2021spectral,mingard2021sgd,mingard2023deep,bhattamishra2022simplicity}.  The success of these models has heralded the end of the ``feature engineering'' era, where features and data representations were hand-crafted by experts \citep{hastie2009elements}, and the rise of ``general-purpose learning algorithms'' that can learn from raw data \citep{lecun2015deep}.

Quantum machine learning (QML) is a newer discipline that aims to leverage quantum computers to produce new algorithms and models that may outperform classical machine learning methods in terms of computational efficiency or accuracy (referred to collectively as \textit{quantum advantage}). While quantum computers are still in the noisy intermediate-scale quantum (NISQ) era~\citep{preskill2018quantum} -- small and prone to errors -- there are many claims in the literature that QML models may already offer some advantage over classical models \cite{abbas2021power,coles_seeking_2021,huang_quantum_2022,schuld2022quantum,bowles2024better}. Variational Quantum Algorithms (VQAs), where a classically parameterised quantum circuit (PQC) is optimised to minimise a cost function evaluated on a quantum computer, are one of the most popular candidates for potential QML advantage \cite{cerezo2021variational}. In this paper, we focus on quantum neural networks (QNNs), a subclass of VQAs in which data is treated as a quantum state and processed by applying a unitary transform in the form of a PQC. The output of the QNN is accessed by measuring the processed quantum state (usually on a designated readout qubit), and the QNN is trained by adjusting the parameters of the PQC to minimise a cost function \cite{schuld_supervised_2021}.

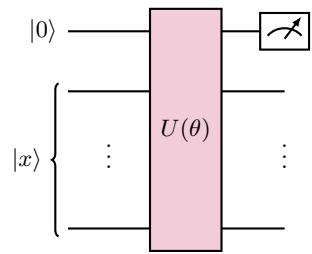
\begin{figure}[H]
    \centering
    \begin{tikzcd}
        \lstick{$|0\rangle$}          & \qw    & \gate[wires=4, nwires=3]{U(\theta)} & \meter{} \\
        \lstick[wires=3]{$|x\rangle$} & \qw    &                                     & \qw \\
                                      & \vdots &                                     & \vdots \\
                                      & \qw    &                                     & \qw
    \end{tikzcd}
    \caption{ \textbf{Schematic of a QNN used for binary classification}. There are $n$ data input qubits encoded in state $|x\rangle$, one $|0\rangle$ initialised readout qubit, and an arbitrary classically parameterised unitary $U(\theta)$. For non-binary data, more readout qubits are required. }
    \label{fig:n_plus_one_wire_circuit_measurement}
\end{figure}

For binary classification tasks, the output is the expectation value of a single Pauli measurement operator on a designated readout qubit \cite{farhi2018classification}. More formally, a QNN being used for binary classification takes in a quantum state $|x\rangle$ encoded on n qubits, and augments it with a $|0\rangle$ qubit, resulting in $|0\rangle |x\rangle$ encoded on $n+1$ qubits (see Figure \ref{fig:n_plus_one_wire_circuit_measurement}). $|x\rangle$ can be thought of as a vector in Hilbert space of dimension $2^n$, and we will often write $x=[x_1, \dots, x_{2^n}]$, the coordinates in the computational basis. Given the $|0\rangle |x\rangle$ input, the QNN then computes the function 
\begin{equation}\label{eq:QNN(x)}
    f(x;\theta) = \bra{x}\bra{0}U^\dag(\theta) Z U(\theta)\ket{0}\ket{x},
\end{equation}
where $U(\theta)$ is a classically parameterised unitary matrix of dimension $2^{n+1}$ and Z is the Pauli measurement operator $\sigma_Z \otimes I \otimes \dots \otimes I$. If $f(x; \theta) > 0$ the class prediction is 1, else 0. $|x\rangle$ might come directly from a quantum system, or it could be classical data $x'$ that has been encoded into a quantum state using a quantum encoding
\begin{equation}\label{eq:phi}
    \phi: x'\rightarrow \ket{x}
\end{equation}

 QNNs can be formulated as kernel methods \cite{schuld2019quantum,havlicek_supervised_2019,schuld_supervised_2021,schuld2021quantum}, where the kernel function for two quantum states $|x^1\rangle$, $|x^2\rangle$ is given by $K(|x^1\rangle, |x^2\rangle) = |\langle x^1 | x^2 \rangle|^2$. If the model is to be used for classical data, then $\phi$ determines the kernel as a function of the classical inputs. This formalism can be used to explore the expressivity and inductive bias of QNNs \cite{kubler2021inductive,huang2021power,abbas2021power,slattery_numerical_2022,liu_rigorous_2021,shaydulin2022quantumkernel}, which has significantly improved understanding of the limitations of QNNs, and the conditions under which they may offer quantum advantage. For example, \citet{huang2021power} show that the exponential feature space of QNNs can make generalization difficult, although there may be quantum advantage for certain kinds of quantum datasets~\cite{,huang_quantum_2022}. Similarly,  \citet{kubler2021inductive} show that generalization (with a polynomial amount of data) is only possible when the data only occupies a small subspace of the feature space (limiting the choice of encoding).
 
While the link to kernel methods has provided important insights into potential QNN performance, it can be abstract and difficult to interpret. 
Here, we aim to provide a more concrete and intuitive understanding of the expressivity and inductive bias of QNNs by exploiting a mapping from QNNs to classical perceptrons acting on the tensor product of the input data.% which we refer to as the \tppn. 
This mapping to a perceptron bypasses challenges in training such as barren plateaus~\cite{mcclean_barren_2018,holmes_connecting_2022}, facilitating the study of generalisation performance and the related inductive bias of QNNs.  
Building on earlier work~\citep{pointing2024do},   
we employ a versatile Boolean dataset to empirically study how QNNs learn across a set of different target functions,  finding that many of the popular methods to encode classical data exhibit expressivity-inductive bias tradeoffs, where more expressive encodings have weaker inductive biases.  We then propose two methods to potentially overcome these limitations; using QNNs to efficiently calculate high-dimensional inner products for use in classical neural network kernels \cite{lee2017deep,matthews2018gaussian},  and constructing a layered non-linear QNN \cite{zhao2021qdnn} that is provably fully expressive on the Boolean dataset, while also exhibiting a richer inductive bias that translates into improved generalisation performance on the Boolean dataset and on a simplified FashionMNIST dataset.  Finally,  we frame prospects for quantum advantage on classical data for QNNs in light of the remarkable performance of DNNs. 

\section{Mapping QNNs onto perceptrons}\label{sec:Quantum_Perceptrons}

In this section, we provide an overview of how to construct the mapping from QNNs to the classical \tppn. A complete formal proof can be found in Appendix \ref{app:perceptron_proof}. We will consider QNNs acting on $\ket{0} \ket{x}$, where $\ket{x}$ is encoded on $n$ qubits, giving a Hilbert space of dimension $N=2^n$.
We can represent $|0\rangle |x\rangle$ as a vector $(x, 0)$ of length $2N$.
Let $x_i$ be the $N$ (in general, complex) components of $x$.

Next, consider a classical perceptron acting on $h \in \mathbb{R}^{N^2}$ which expresses functions of the form
\begin{equation} \label{eq:perceptron}
    g(x) = w \cdot h
\end{equation}
where $w\in \mathbb{R}^{N^2}$ are the weights. 
%Now consider an embedding of the form$$ \psi: x \rightarrow x\ostar x, $$ 
A  QNN  defined by any unitary matrix $U$ acting on $(x, 0) \in \mathbb{R}^{2N}$ as in \cref{eq:QNN(x)} can be uniquely mapped onto the perceptron defined in \cref{eq:perceptron} with  $h =x \ostar x $ (see \cref{proof:tpp=qnn}), where we define the complex tensor product $x \ostar x \in \mathbb{R}^{N^2}$, as
\begin{align}\label{eq:perceptron_equivalence}
    (x \ostar x)_{N(i-1)+i} &= x_i x_i^* \\
    (x \ostar x)_{N(i-1)+j} &= Re(x_i x_j^*) + Im(x_i x_j^*) \;\;\; \text{for } i \neq j
    % &[x_i x_i^*, Re(x_i x_j^*)+Im(x_i x_j^*),\notag \\
    % &Re(x_j x_i^*)-Im(x_i x_j^*), x_j x_j^*, ...].
\end{align}
When all components of $x$ are real, $x\ostar x$ is equivalent to the standard Kronecker product $x \otimes x$.
We will refer to this perceptron with real weights $w$ acting on $x \ostar x$ as the tensor product perceptron (\tppn). In other words, any QNN can be uniquely mapped onto a TPP.  A preliminary observation of this link between a perceptron and a general QNN can be found in  Eq.\ 45 of \citep{schuld_supervised_2021}) in the context of a more general link between QNNs and kernel methods.  

It is also possible, under certain mild conditions on the weights of the TPP, to construct a mapping from \tpps to QNNs such that the numerical output is the same. If we are only interested in the sign of the output (as we are for classification), then the models have equivalent expressive power (see \cref{proof:qnn=tpp}).

\section{Expressivity of TPPs on Boolean functions}
In this section, we investigate the expressivity of \tpps (and therefore QNNs) on Boolean functions, which are foundational building blocks of computer science.
For the $n$-bit Boolean dataset, inputs are of the form $x \in \{0,1\}^n$, and the target function can be expressed as a binary function $f: \{0,1\}^n \rightarrow \{0,1\}$, where the inputs are ordered by ascending binary value (see \cref{sec:Boolean_system}). For example the $n$-dimensional parity function  is defined as $f(x) = \sum_{i=1}^n x_i \mod 2$. There are $2^n$ inputs and thus $2^{2^n}$ possible target functions. The Boolean dataset is a useful testbed for machine learning algorithms because it can be used to measure the properties of a model across a range of target functions. This contrasts with many image or text-based datasets where there is just one target function.  Moreover, the long tradition of research on  Boolean functions provides opportunities to derive analytical results.

The perceptron, introduced in 1958 by Rosenblatt \citep{rosenblatt_perceptron_1958} as a candidate for a general machine learning algorithm, is a form of linear classifier.  In their famous book, \citet{minsky1969introduction} demonstrated its inability to express basic non-linear functions such as XOR.  More generally, the number $T(n)$ of Boolean functions that a perceptron can express is bounded by $2^{n^2 - n \log_2 n - \mathcal{O}(n)} \leq T(n) \leq 2^{n^2 - n \log_2 n + \mathcal{O}(n)}$~\cite{vsima2003general}, which is exponentially fewer than the total number of Boolean functions on this data. This lack of expressivity, together with the difficulty at the time of training deeper networks, has been cited as one of the causes of the first AI winter \cite{aiwinter}.

In that context, it is interesting to examine the expressivity of the \tpp (which can also be viewed as a form of quadratic classifier on $x$). If $x$ has dimension $N$, and so $w$ has dimension $N^2$, the number of Boolean functions $T(N)$ the \tpp will be able to express (see \cref{lemma:num_Boolean_functions_expressible}) is upper-bounded by 
\begin{equation} \label{eq:expressible}
       T(N) \leq  2^{N^3+N^2\log_2(e)+1}
\end{equation}
The \tpp can only express an exponentially small fraction of the  $2^{2^N}$  possible functions on these data points.  While the \tpp is more expressive than the standard perceptron and can express the basic $n=2$ parity function, XOR, it cannot express any further parity functions with $n \geq 3$ (see \cref{sec:sec:qnn_amplitude_parity_proof}).

This analysis implies that, much like the basic perceptron, the \tpp is a severely restricted model with a limited range of expressible functions.   The equivalence between QNNs and \tpps means that the same problems carry over for QNNs.  To use a QNN or a \tpp to classify some new dataset, one has to ensure that the target function is expressible, which may mean performing explicit feature engineering to ensure that it is.  This situation harkens back to earlier days of machine learning where feature-engineering played an important role.    By contrast, DNNs are highly expressive~\cite{hornik1989multilayer}, and even simple FCN architectures don't need to be that big to express every Boolean function~\cite{mingard2019neural}.

\section{Experiments on the Boolean dataset}\label{sec:qb:exp_bool}

\begin{figure*}[ht]
    \centering
    \includegraphics[width=0.9\textwidth]{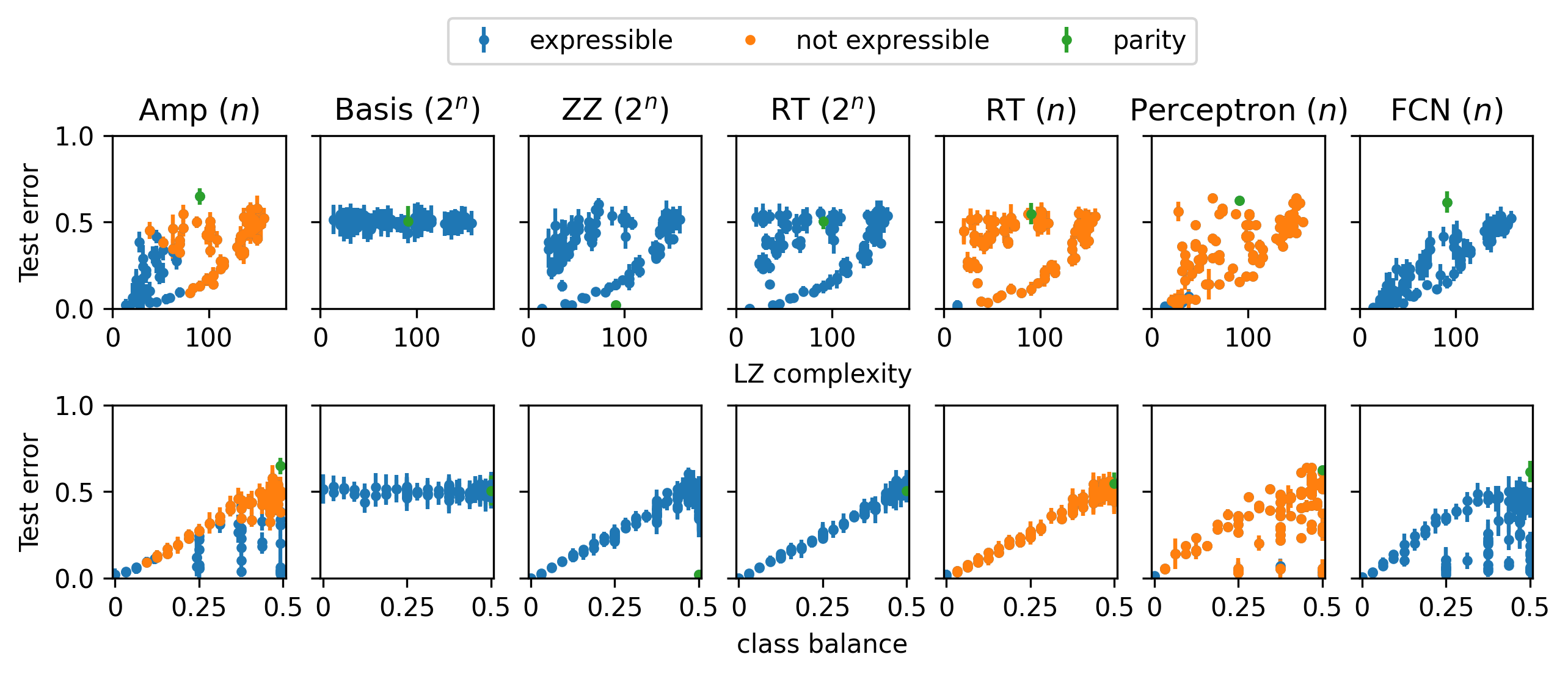}
    
    \caption{\justifying
    \textbf{Test error v.s.\ LZ complexity and class balance for different encoding methods \& algorithms for supervised learning on the  $n=7$ Boolean dataset.} The first five columns show QNNs with different encoding methods and the two final columns show two classical neural networks -- the perceptron and a 1-hidden layer FCN. The bracketed values after each encoding type give the dimension of the classical (quantum) input (Hilbert) space.
    The top (bottom) row shows generalisation error v.s.\ LZ complexity (class balance). Class balance is the minimum proportion of 0s or 1s in the function. 
    Each datapoint is for one of 100 target functions chosen to have a wide range of entropies and LZ complexities. The perceptron can fit 7 functions and the QNN with amplitude encoding can fit 41 (neither can fit the parity function), while the QNN with RT(n) encoding can only fit the trivial function. All other learning algorithms can express all 100 target functions.
    The QNN with basis encoding has no inductive bias, and the QNN with either RT encoding has a simple bias towards low class balance functions, performing poorly on functions with high class balance and low complexity. ZZ encoding has a similar inductive bias, except it achieves a perfect 0\% test error on the parity function, for which it has been engineered. The QNN with amplitude encoding and the perceptron are very similar in their inductive biases towards low class balance and towards low LZ complexity.  The 1-hidden layer FCN is even more biased towards simple (low LZ complexity)  functions, and is fully expressive -- a clear upgrade on both fronts over the perceptron and the QNNs.
    }
    \label{fig:qnn_posterior}
\end{figure*}

In this section, we empirically study the expressivity and inductive bias of QNNs with different encodings on Boolean data.   Given a subset of input-output relations for a target Boolean function, a key challenge for supervised machine learning is to predict the unknown input-output pairs.  Experiments are performed on the $7$-bit Boolean dataset, which has 128 input data points and $\approx 3 \times 10^{38}$ possible target functions.
Four different encoding methods $\phi$ for $\{0,1\}^n$ are examined,  all of which have very different properties. They are: 
\begin{enumerate}
    \item Amplitude encoding $\phi_A$ encodes classical data into the amplitudes of a quantum state, using $\lceil\log_2(n)\rceil$ qubits to represent $n$ dimensional classical data as 
    $$
    \phi_A(x)=\frac{1}{\norm{x}} \sum_{i=1}^{2^n} x_i \ket{i}
    $$
    For Boolean data, the datapoints form a (normalised) hypercube in the n-dimensional Hilbert space.
    This encoding is most similar to the encoding for classical neural networks. Due to the normalisation constraint, the point at the origin in the classical encoding cannot be encoded, so it is dropped. See \cref{def:enc:amp}.
    \item Basis encoding $\phi_B$ uses $n$ input qubits, such that $[0,1,0]\rightarrow \ket{0}\ket{1}\ket{0}$. See \cref{def:enc:basis}.
    \item ZZ encoding $\phi_Z$ uses $n$ input qubits, and is designed to be biased towards the parity function. See \cref{def:enc:zz}.
    \item Random Transform encoding $\phi_R$ uses either $n$ or $\lceil\log_2(n)\rceil$ input qubits, and is the output of a single ReLU-activated fully connected layer with either $n$ or $2^n$ outputs (we refer to these as RT($n$) and RT($2^n$) respectively). It is designed to act as a random non-unitary transform. See \cref{def:enc:relu}.
\end{enumerate}

Each encoding method embeds data very differently in Hilbert space. We test how these embeddings perform on 100 target functions, chosen to have a wide range of different characteristics.
These functions include the parity function, functions with fixed class balance, functions with different degrees of symmetry, and entirely random functions (see \cref{app:targ_train_details}). 
We measure the complexity of these functions using two methods: class balance, and Lempel-Ziv (LZ) complexity. Class balance is defined as min$\{p, 1-p\}$, where $p$ is the proportion of data points classified as $0$ (and so $1-p$ are classified as 1).  LZ complexity is measured as in~\cite{valle2018deep} using the Lempel-Ziv 76 compression technique, and is low for highly structured data, and/or data with low class balance. By comparing both these measures we can distinguish between learners that have an inductive bias towards highly structured data, and those that simply learn to predict the frequencies of $0$ or $1$  in the training data. Previous work has shown that DNNs display a strong inductive bias towards functions with low LZ complexity, even when class balance is high, a property thought to be critical to their ability to generalize well in a wide range of practical situations. \citep{mingard2023deep,valle2018deep,mingard2019neural}.

We trained the QNN (using the \tpp correspondence above) on a randomly selected training set of 64 data points from each function (testing on the other 64). See \cref{app:targ_train_details} for details on the functions and the training process.
We repeated this for 5 random seeds per function, with a different random training set each time. Training the \tpp is much more efficient than QNNs and we can always reach low loss, as there are no barren plateaus. We then used satisfiability tests on the whole dataset to determine whether the QNN can express each of the 100 functions. We also repeated these experiments on a classical perceptron, and a 1-hidden layer FCN of size $(n,2^n,1)$ that is provably able to express all possible Boolean functions \cite{mingard2019neural}.

In \cref{sec:quantum_kernels} we present similar experiments that explicitly use the kernel correspondence of QNNs~\cite{schuld2019quantum,havlicek_supervised_2019,schuld_supervised_2021,schuld2021quantum}.   By directly examining the eigenvalues of the kernels as well as the task-model alignment measure from~\citet{canatar2021spectral}, we can explain why the simple FCN kernel outperforms QNNs on this data. This kind of analysis which is standard in the literature on kernels \citep{cui2021generalization,simon2021theory,harzli2021double} may be better suited for exploring the potential promise of new QNN architectures than simply comparing generalisation performance, as is typically done in the literature~\cite{schuld2022quantum,bowles2024better}. 

\subsection{Expressivity and inductive bias of amplitude encoding and RT(n) encoding}

 A QNN using amplitude encoding can only express an exponentially small fraction of all Boolean functions of the data, as the QNN has an $n$ dimensional Hilbert space for $n$-bit Boolean data. Indeed, we find that this QNN can only express 41 of our 100 functions (see \cref{fig:qnn_posterior}). This means that when applying this model to new datasets a detailed understanding of the data distribution will be required to ensure that the encoded data does not contain structures that prevent the potential target function from being expressed. It is likely hard to prevent this problem \textit{a-priori} for data from quantum sources where the structure may be subtle and hard to predict.  % This hearkens back to the days of ``feature engineering'', and is a severe limitation on a general learning algorithm.

While there are hard constraints on the expressivity of QNNs with amplitude encoding, they do exhibit an inductive bias towards structured functions on the Boolean dataset. This can be seen in \cref{fig:qnn_posterior}, where we observe that amplitude encoding has an inductive bias towards highly structured data (as measured by the Lempel-Ziv complexity, top), even for some functions with high class balance.

The RT(n) encoding with $\lceil \log_2{n} \rceil$ qubits has no structure due to the random nature of the transform. This is reflected in its inability to generalise well for any structured high class balance but low LZ complexity functions. The QNN is also very inexpressive, only able to fit the trivial function to zero training error. 

\subsection{Inductive bias of fully expressive QNNs}\label{sec:ind_bias_other_encodings}

Each of the following encodings uses $n$ qubits to encode $n$-dimensional classical data, and achieves full expressivity on the $n$-bit Boolean dataset. In this context, we note that if the encodings use all $2^n$ dimensions, then \citet{kubler2021inductive} show that an exponential amount of data would be required to learn the data, implying no quantum advantage. To learn from a reasonable (polynomial) amount of data, one would need to significantly restrict the expressivity of the QNNs, using only a limited number of dimensions in Hilbert space to encode the data. See \cref{sec:lit_review} for more details.  How these principles work out can be partially illustrated below.

With basis encoding it is not hard to prove, see for example~\citep{ngoc2020tunable},  that all $2^{2^n}$ functions can be expressed by a QNN with $2^n$ gates. However, since every datapoint is encoded orthogonally in a $2^n$ Hilbert space, the kernel is trivial, and the QNN has no inductive bias.  As can be seen in \cref{fig:qnn_posterior},  the test error for every function is 50\% illustrating that for basis encoding the QNN can't learn Boolean data.

\cref{fig:qnn_posterior} shows that ZZ encoding has a strong bias towards low class balance - its performance is equivalent to simply measuring the frequency of 1s or 0s, and using this to predict unknown data. More interestingly, it also shows a further strong bias towards the parity function, with perfect generalization for $n=7$ and with $64$ training data. 
The prediction from \citep{kubler2021inductive} is that ZZ encoding would fail to generalise for increasing $n$ (assuming $\phi_{ZZ}(x')$ continues to span an exponential number of dimensions for higher $n$).  This dataset is most likely too small to test such predictions.  Nevertheless,   we used $2^{n-1}$ datapoints for training, out of $2^n$, and the prediction is that training sets will need to continue to grow exponentially with increasing $n$ for this system to generalise,  even for the parity function. 

Note that there is nothing particularly deep about the ZZ encoding being good at learning the parity function.  For a given target function, an encoding can always be created with the desired inductive bias.  
 While the 1-layer FCN we use does not perform well on the parity function with classical encoding, other architectures do have a better inductive bias towards parity (see experiments in \citep{mingard2023deep}). It is not hard to create an encoding more optimal for the parity function  for classical DNNs as we show in  \cref{app:zz}. In summary, designing a specific encoding to allow a QNN to work well on a specific target is not quantum advantage. Although such claims are frequently made in the literature, we advocate for this practice to be tempered.   The goal of classical general machine learning algorithms such as DNNs is to learn features (an embedding) that allow it to separate the data itself without the need for tuning by hand.

\cref{fig:qnn_posterior} shows that, in contrast to the RT($n$) encoding, the RT($2^n$) encoding is fully expressive on this dataset.  Moreover, much like the ZZ encoding, it has an inductive bias towards low class balance, but no bias towards the more interesting structured Boolean functions.

These three encodings illustrate a general point: with data spread out across all $2^n$ dimensions of Hilbert space, any function can be expressed, but the inductive bias will be weak. This tradeoff makes it hard to design  QNNs that function like DNNs, that is as general machine learning algorithms.

\section{Quantum DNNs \& Quantum FCN Kernels}\label{sec:Quantum_DNNs}

While one could continue trying new embeddings, and hoping for better performance, the problem demonstrated in the previous section and our more general experience with classical kernel methods call for new strategies.  One way forward we explore is to see how a QNN could mimic an FCN kernel on Hilbert space -- providing an easy way to improve the expressivity and inductive bias of quantum kernels -- but still suffering from all the problems with kernel machines discussed in the context of QNNs.
The other strategy moves away from the kernel limit with a deep quantum neural network (DQNN) that is capable of expressing non-linear functions. 

\subsection{FCN kernels}\label{sec:QNNK1}

Rather than directly training a QNN, the quantum kernel $K_Q(x^{(1)},x^{(2)})=|x^{(1)\dag}x^{(2)}|^2$ can be estimated on a quantum computer, and used to perform inference classically. One can then use tricks such as 
 a `bandwidth' parameter to artificially move inputs closer together which can counteract the effects of data being embedded in too high-dimensional a space~\cite{canatar2022kernel}. This may allow kernel learning with polynomial data for systems that would not have been otherwise learnable~\cite{kubler2021inductive}.
Here, we show another simple twist on standard quantum kernel methods.  Classical DNNs in the infinite-width limit can be mapped to kernel methods,  and these have been shown to be competitive with finite DNNs on small data sets~\cite{arora2019harnessing,lee2020finite}.
For $l$-layered classical FCN ReLU networks with biases set to 0, we can recursively calculate such kernels (in the infinite width limit) $K^l_{ij}=K^l(x^{(i)}, x^{(j)})$ using the following recurrence relation \citep{lee2017deep,matthews2018gaussian},
\begin{equation}\label{eq:fcn_kernel}
a_l^{(ij)} = \frac{1}{\pi}
\left(\sqrt{1 - (a_{l-1}^{(ij)})^2} + \left(\pi - \arccos{a_{l-1}^{(ij)}}\right)a_{l-1}^{(ij)}\right).
\end{equation}
where $a_l^{(ij)} \equiv K^l_{ij}/\sqrt{K^l_{ii}K^l_{jj}} $ and $a_0^{(ij)}=x^{(i)}\cdot x^{(j)}$. We can construct a kernel $K_Q^l(x^{(i)}, x^{(j)})$ which mimics the kernel of an $l$-layered infinite-width FCN on Hilbert space  by sampling   from a QNN \citep{canatar2022bandwidth},  choosing $a_0^{ij}=\sqrt{K_Q(x^{(i)}, x^{(j)})}$ and  using  \cref{eq:fcn_kernel} to construct $K_Q^l$.  In our experiments we only use $l=1$. While this method does not necessarily improve on classical kernels (for example there is still no representation learning)  it gives us access to a different kernel than the standard QNN kernel $K_{Q}=|x^{(1)\dag}x^{(2)}|^2$.    

We tested this model on the $n=7$ Boolean dataset. 
$K_Q^1$ (for a 1 hidden layer FCN) performs in general slightly worse than the FCN on the Boolean dataset in \cref{fig:dqnn_posterior}. But, it performs better and is more expressive and/or has a more interesting inductive bias than QNNs with the encodings in \cref{fig:qnn_posterior}.

We also tested performance on a version of the FashionMNIST~\cite{xiao2017fashion} image dataset that is popular in the QNN literature.  To be learnable by QNNs, this dataset is typically hugely simplified.  Here it is produced by binarising the FashionMNIST labels, performing principal component analysis on the entire dataset, and taking the top 8 components.  These are then encoded with amplitude encoding on 3 qubits. This kind of downsampling is popular in the  QNN literature, and can be found in standard set-ups such as Tensorflow Quantum~\cite{broughton_tensorflow_2021}. Unfortunately, we found multiple instances in the QNN literature where the extreme simplification of these datasets is not well signposted, see also \citep{bowles2024better}. We discuss this problem in~\cref{app:over-simplified}.  As one step forward we will call this simplified version Q-FashionMNIST, even though very little of the original FashionMNIST dataset structure is left.   We use 250 training datapoints and 50 test datapoints. %This construction is very similar to the dataset used by TensorFlow Quantum in \cite{TensorFlow_2023} and \citep{dborin2022matrix} and massively simplifies the dataset (to be runnable on a QNN simulation). We feel that these simplifications are often not made explicit enough in QML papers.
Interestingly, the $K_Q^1$ kernel achieves the best accuracy of all the models we tested on Q-FashionMNIST in \cref{fig:dqnn_fashion}.  

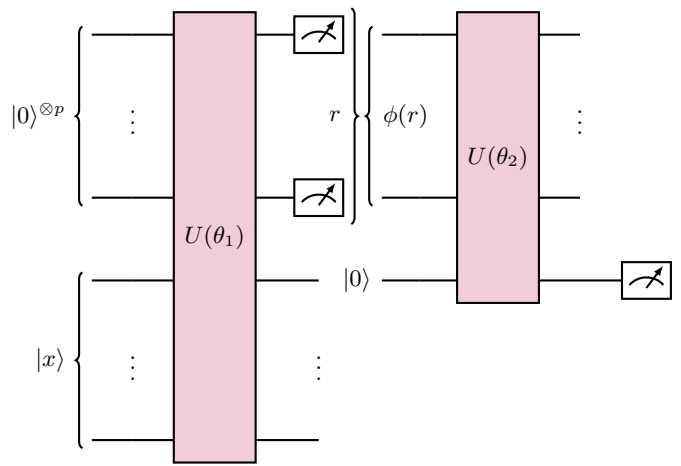
\begin{figure}[t]
    \centering
    \begin{tikzcd}
        \lstick[wires=3]{$|0\rangle^{\otimes p}$} & \qw    & \gate[wires=6, nwires={2,5}]{U(\theta_1)} & \meter{} \rstick[wires=3,label style={xshift=0.2cm}]{$\phi(r)$} & \lstick[wires=3,label style={xshift=-0.2cm}]{$r$} & \qw & \gate[wires=4, nwires=2]{U(\theta_2)} & \qw & \\
                                                  & \vdots &                                         &  & & & & \vdots & \\
                                                  & \qw    &                                         & \meter{} & & \qw & & \qw & \\
        \lstick[wires=3]{$|x\rangle$}             & \qw    &                                         & \qw & \lstick{$|0\rangle$} & \qw & & \qw & \meter{} \\
                                                  & \vdots &                                         & \vdots & & \\
                                                  & \qw    &                                         & \qw & &
    \end{tikzcd}
    \caption{\textbf{A DQNN used for binary classification.} There are $n$ data input qubits and $p$ $\ket{0}$ initialised intermediate qubits in the first layer. The measured output of the first layer $\langle r \rangle$ is then re-encoded into a quantum state  $\ket{r}$ (using an encoding method $\phi$), and passed through a final layer with one $\ket{0}$ initialised output qubit.}
    \label{fig:dqnn_plus_one_wire_circuit_measurement}
\end{figure}

\subsection{DQNNs}

To escape the kernel paradigm, some non-linear operation has to occur during the forward pass. Other existing quantum machine learning methods can achieve this. For example, convolutional neural network inspired QCNNs~\cite{cong_quantum_2019} do this via qubit measurement, and proposed continuous variable algorithms use `non-Gaussian' gates.
A classical FCN is made up of a series of perceptron-like layers. By analogy, we create a deep QNN (DQNN) by including a measurement and an extra layer, as shown in \cref{fig:dqnn_plus_one_wire_circuit_measurement}.  The first layer is a standard QNN, but with $p$ readout qubits. The second layer is a standard QNN with $q$ input qubits (the exact value will depend on the encoding used in the second layer), and one readout qubit. A  similar DQNN with layered structure and measurements in the middle can be found in~\cite{zhao2021qdnn}.   
The data is encoded for input to the first layer $x^{(L=0)}=\ket{0}^{p}\ket{x}$. The output of this layer is then a vector $z$ of dimension $p$ with components $z_i$ -- the expectation values of the $i$'th readout qubit. The output of the first layer is $[\langle q_0 \rangle, \dots, \langle q_p \rangle]$. A trainable (classical) bias term $b_i$ can be added to each of these, followed by a non-linearity $\sigma$, to give $x^{(L=1)}=[\sigma(\langle q_0 \rangle+b_0), \dots, \sigma(\langle q_p \rangle+b_p)]$.
This output can then be encoded for the second layer with some encoding method $\phi^{(L=1)}(x^{(L=1)})$, and the second layer of the DQNN classifies this input like a regular QNN.  
In  \cref{app:sec:dqnn_proofs}  we demonstrate that this method can model a 1-hidden layer FCN with step-function activations, and therefore satisfies universal approximation theorems.

\begin{figure}[t]
    \centering
    \includegraphics[width=1\columnwidth]{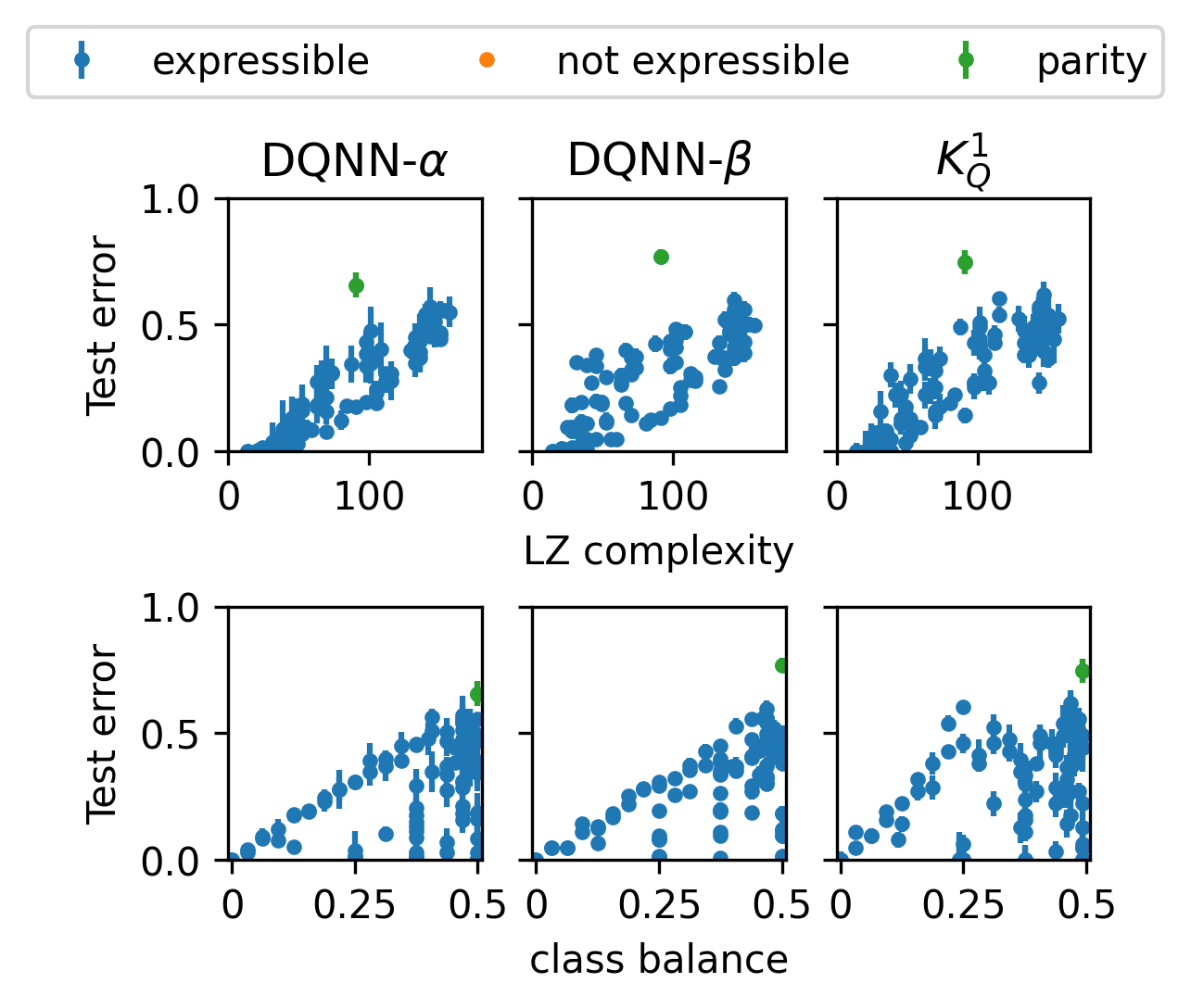}
    \caption{\justifying\textbf{Test error v.s.\  complexity and class-balance for DQNNs and $K_Q^1$ for the $n=7$ Boolean dataset} 
    The top (bottom) row shows generalisation error v.s.\ LZ complexity (class balance). %Class balance is the minimum proportion of 0s or 1s in the function.
    Each datapoint is for one of the 100 target functions used in ~\cref{fig:qnn_posterior}. %chosen to have a wide range of entropies and LZ complexities.
    All three algorithms receive data with amplitude encoding and are fully expressive. The \dqnna (with $q=6$ intermediate qubits) has the most similar inductive bias to the kernel $K_Q^1$ and the finite-width FCN (see \cref{fig:qnn_posterior}). The inductive bias of \dqnnb ($q=7$) is weaker than \dqnna due to the basis encoding used after the first layer. 
    }
    \label{fig:dqnn_posterior}
\end{figure}

\begin{figure}[h]
    \centering
    \includegraphics[width=0.8\columnwidth]{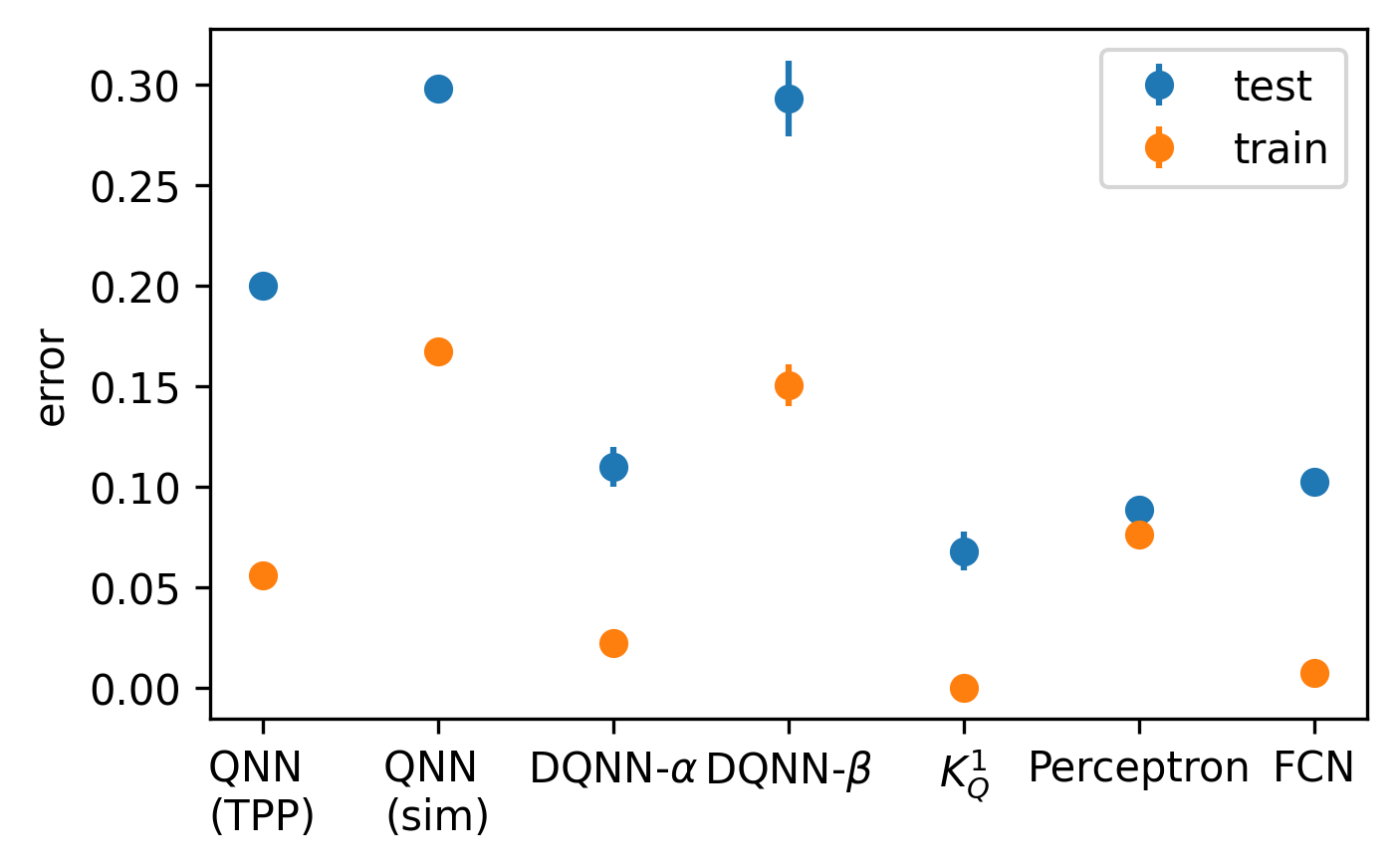}
    \caption{\justifying\textbf{Generalisation error for different architectures on the simplified $8$-component Q-FashionMNIST dataset, amplitude encoded on 3 qubits.}
    The QNN (sim) and QNN (TPP) are amplitude-encoded QNNs simulated with  PyTorch and the TPP classical mapping respectively until their train accuracy no longer improves. The QNN (sim) does not fully converge due to barren plateaus and so has a higher training error.  
    As in \cref{fig:qnn_posterior}, the \dqnna outperformed the \dqnnb. % Both have $q=3$ intermediate qubits.
    $K_Q^1$ performs better than the FCN. The simple perceptron converges to a higher training error than the TPP, but a lower test error because it has a better inductive bias for this problem. 
    }
    \label{fig:dqnn_fashion}
\end{figure}

In our experiments, the nonlinearity $\sigma$ is thresholded at 0, and we use two encodings $\phi^{(L=1)}$ in the second layer
\begin{enumerate}
    \item DQNN-$\alpha$ networks use amplitude encoding -- for $p$ readout qubits in the first layer, $x^{(L=1)}$ is encoded onto $q=\lceil \log_2 p\rceil$ qubits in the second layer using amplitude encoding.
    \item DQNN-$\beta$ networks use basis encoding -- for $p$ readout qubits, $x^{(L=1)}$ is encoded onto $q=p$ qubits in the second layer using basis encoding
\end{enumerate}

We trained the DQNN  with backpropagation using PyTorch \citep{paszke2019pytorch}.  
\cref{fig:dqnn_posterior} compares these  DQNNs to a 1 hidden-layer FCN and a perceptron on our Boolean dataset.
Both DQNNs perform better than the QNNs in \cref{fig:qnn_posterior}, and \dqnna performs better than \dqnnb.

In 
\cref{fig:dqnn_fashion} we compare test performance of the DQNNs with a QNN, the hybrid quantum-classical kernel $K_Q^1$, a perceptron and an FCN on the downsampled image dataset Q-FashionMNIST.
We also compare a QNN trained using the TPP correspondence (where we could reach the global minimum) to a QNN trained directly with PyTorch until the loss plateaus for several hundred iterations. In contrast to the TPP, we were unable to reach the global minimum due to what appear to be barren plateaus. This optimisation problem is responsible for a 10\% drop in both train and test accuracy compared to the TPP. 
Interestingly, while the \dqnnb does not perform well, 
the  \dqnna (with $p=8$ intermediate readout qubits, so the last layer had $q=3$ data input qubits) performs comparably to the simple FCN, reaching $98\%$ training accuracy (against $99\%$ for the FCN).  The \dqnna can express non-linear functions since it trains to lower error than the perceptron (which has $92\%$ training accuracy).  However, it generalises slightly less well, indicating that its inductive bias is less good than that of the perceptron. This difference is perhaps not surprising given that the TPP kernel is the square of the perceptron kernel and so has weaker correlations which translates here to a weaker inductive bias.

While they perform better than QNNs, the exact type of DQNNs we propose here would be computationally inefficient when scaled up. First, a forward pass would be very expensive, as expectation values need to be taken at each layer. Second, if we are using amplitude encoding on the second layer as for \dqnnan, a large number of readout qubits would be needed, and the data would then have to be encoded on the input qubits of the next layer, another expensive process. Both of these problems can be somewhat alleviated, for example by measuring input qubits rather than having separate output qubits, or using probabilistic outcomes at the hidden layers. However serious computational efficiency problems remain. % it illustrates a point: as making perceptrons multi-layered massively increases their expressivity without damaging their inductive bias, adding a layer to a QNN has an analogous effect.

% \FloatBarrier
%\input{sections/conclusion}
\section{Comparing QNN and DNN performance}

So far, we have used the relative simplicity of the \tpp architecture to analyze the expressivity and inductive bias of QNNs.   In this section, we place these results for QNNs into the wider context of classical DNNs, the most successful general-purpose learning algorithms currently available. 
We ask: what barriers do QNNs face to also be useful as general-purpose learning algorithms?  
 
 Many of the points in this section may seem quite obvious to those versed in classical deep learning. Some have already been discussed in the QNN literature (see e.g.~\cite{schuld2022quantum}).  Nevertheless, our experience of reading the recent QNN literature suggests that these ideas have not yet been assimilated.  We therefore provide a more expansive discussion of than one might otherwise do in the hope that this will help move things forward. 

\noindent \textbf{Mismatched scales complicate the search for quantum advantage with current QNNs. }
Classical machine learning has benefited from many more years of study and investment than quantum machine learning has.  Moreover,  the power of deep learning techniques only fully emerges on large datasets and with extremely large models. For example, the modern deep-learning era is often said to have kicked off in 2012 when a CNN-based architecture called Alexnet~\cite{krizhevsky2012imagenet} with 62 million trainable parameters won, by some distance,   a competition on ImageNet-1K, a database with $1.2 \times 10^6$ training images and $10^3$ classes. Recent large language models (LLMs) such as GTP4, based on transformer architectures, can have over $10^{12}$ parameters and be trained on close to a terabyte of data \cite{openai2024gpt4}. These scales are many orders of magnitude larger than current QNNs, and without major breakthroughs, any QNNs in the foreseeable future.

From the perspective of mismatched scales, directly comparing empirical QNN performance to classical DNNs does not seem like a well-posed question at this point in time.  Nevertheless, as explained for example in \citet{schuld2022quantum}:  \textit{The dominant [current] goal in quantum machine learning is to show that quantum computers, with their properties like entanglement and interference, offer advantages for machine learning tasks of practical relevance \ldots This question is particularly important to the emerging quantum technology industry which has been driving the ``tunnel vision of quantum advantage''}.   Indeed, many papers in the literature report an empirical quantum advantage on classical data, albeit on very small datasets.  From the vantage point above, the popularity of this direction of research is surprising, as are the results.   Unfortunately,  many of these claims fail upon closer inspection. For example,  as recently shown in some detail by~\citet{bowles2024better},  simple off-the-shelf DNNs typically outperform their QNN counterparts.     We independently found similar problematic issues in the QNN literature that obscure the extent of the challenge to achieve quantum advantage.  We detail two in \cref{app:abbas_paper}: the use of extremely down-sampled datasets that are not clearly described, and claiming quantum advantage while making comparisons to sub-standard classical DNNs.    We join with the authors of~\cite{bowles2024better} to call for a change of research practices in the field.

\noindent \textbf{Lack of theoretical understanding of classical DNNs makes extrapolation from tiny datasets problematic}
Despite the caveats above,  the question of whether QNNs can at some time in the future achieve significant quantum advantage on the kinds of classical data that underlies commercially relevant applications remains an important one.  The enormous investments of time and money that will be needed to bring this promise, if it exists, into fruition injects extra urgency into this question.

One way forward is empirical, to study the performance of QNNs on the tiny datasets that they can currently learn in the hope that one can extrapolate out to predict how they may perform on larger datasets. Unfortunately, inferring potential quantum advantage is hugely complicated by our limited theoretical understanding of why new behaviours, not evident in small datasets, emerge with scale in classical DNNs. Without clearer theoretical guidance, making reliable extrapolations about quantum advantage from empirical QNN performance on small datasets is extremely difficult. Even if a QNN outperforms a classical DNN on a tiny dataset, that advantage may not persist as the dataset size increases. Conversely, if a QNN underperforms compared to a classical DNN on a small dataset, it doesn't necessarily mean that scaling up the QNN architecture won't yield significant advantages on larger datasets.

\noindent \textbf{What can be learned from tiny datasets? }
Nevertheless, there still remain interesting questions that can be empirically interrogated with smaller datasets.  The mystery of why DNNs perform so well can be split into two distinct questions~\cite{mingard2023deep}:  1) Why don’t overparameterised DNNs overfit as classical learning theory predicts that they should?  2) Given a DNN that does not overfit, what hyperparameter/architecture combinations optimise performance on large datasets?  Much current research in the field is aimed at question 2, since one can simply observe that DNNs don't overfit and take that as given. Then the key questions,  which matter most in engineering practice anyhow, revolve around how to tune hyperparameters or architecture to enhance performance.   In recent years, much theoretical progress has been made in answering question 1, and kernel methods have played an important role in this quest (see e.g.~\cite{belkin2021fit} for a review). Nevertheless,  this novel understanding has yet to translate into consensus theories that can be used systematically answer question 2, including explaining how and why DNN performance improves with scale. 

The Boolean dataset we discuss in section~\ref{sec:qb:exp_bool} is particularly well suited to investigating question 1~\cite{valle2018deep,mingard2019neural,mingard2023deep}.  %An advantage it has over many of the down-sampled image datasets typically used in the field of QNNs is that 
One can easily vary 
the complexity of the target function in order to systematically explore the inductive bias of a learning system. Facilitated by our mapping to classical TPPs,  we show in Fig 2 that a series of popular embeddings for QNNs exhibit either a rather uninteresting inductive bias (for example ZZ and basis encoding) or suffer from a lack of expressivity (amplitude encoding).   A minimal requirement for QNNs would be that they show similar inductive bias and expressivity to classical DNNs. The fact that they don't here suggests that it may be harder than perhaps anticipated for current QNNs to overcome even this 1st order hurdle. 

 \noindent \textbf{Using theory to extrapolate to larger scales}
While the impressive performance of large-scale DNNs is still theoretically mysterious, QNNs can be mapped onto the much better-understood framework of kernels~\cite{schuld_supervised_2021,kubler2021inductive,huang2021power}.  This provides a more promising basis for making theoretical predictions about the performance of larger QNNs (see also \cref{sec:lit_review} and \cref{sec:quantum_kernels}).   On the other hand, the fact that QNNs map to a restricted class of kernels raises another rather obvious note of caution about any potential performance advantage for QNNs on classical data.   It is well known that DNNs have so far significantly outperformed kernel methods at scale, despite DNN-inspired kernels being competitive on smaller datasets~\cite{arora2019harnessing,lee2020finite}.

The mapping to the \tpp provides further specific theoretical insights beyond the more general ones that kernels bring.    Firstly, the \tpp suffers from the classical expressivity problems of linear classifiers.  One might argue that this critique is unfair since all kernel methods can be mapped onto a linear classifier with a feature map. 
However, in contrast to QNNs, where the feature map is explicitly constructed via the embedding, the power of classical kernels comes from the kernel trick where the feature map is treated implicitly.  This greatly simplifies calculations because the feature map -- which can be infinite -- doesn't need to be directly calculated, and instead one works directly with the kernel.   QNNs map to a restricted set of finite-sized feature maps. For some feature maps, such as amplitude encoding, all expressivity issues carry over.  
The worry is that other feature maps might also encounter unlearnable function problems unless handled carefully. For quantum data, where creating custom embeddings is more challenging, avoiding these expressivity issues may be hard to avoid.  

Of course, with sufficient knowledge of the target function and data distribution, one can always design specialised embeddings that work well on specific data. But, that should not be called quantum advantage unless one can show that the same cannot be achieved with  classical kernels or  DNNs.   Moreover, this paradigm of hand-encoded features has been largely abandoned with the advent of large DNNs that automatically extract features and which routinely outperform hand-coded methods~\cite{lecun2015deep}. Moreover, these DNNs are versatile, and don’t need to be significantly adapted for every new problem.  
In summary, the fact that QNNs depend critically on the details of the quantum embedding, or in classical language, a hand-crafted feature map, is a disadvantage that experience with modern machine learning suggests will be hard to overcome.

\noindent \textbf{Reducing the expressivity of QNNs.}
The problems that the mapping to the TPP exposes hold for any fully expressive QNN and cannot be easily sidestepped.  However, one can restrict the expressivity of a QNN, a route that has been explored by many authors. In classical machine-learning theory restricting the expressivity of a learner is an important factor in optimising bias-variance tradeoff.  If QNNs are best understood within this older classical setting, then reducing expressivity may be necessary to improve performance.  We explore this idea briefly in \cref{sec:charlie}, showing that reducing expressivity does not significantly improve performance on our Boolean dataset.  While more work would need to be done to firmly establish this trend,  modern techniques such as deep learning and DNN-inspired Gaussian processes have shown that the best performance typically comes from generalist models with high capacity.  This experience suggests that restricting expressivity, which typically needs to be optimised for each dataset, is unlikely to produce a general-purpose QNN that can rival modern DNNs on classical data. 

\noindent \textbf{Scaling up and the TPP correspondence.}
Our TPP correspondence provides a more efficient way to simulate many current QNNs than direct methods do.   But this advantage will likely disappear as technology improves so that QNNs can handle larger input dimensions.    For example, a task with $n$ dimensions and binary data has an exponential input size of $2^n$ that will quickly become impractically large for TPPs.  At first sight, it appears that QNNs can circumvent this problem because just $n$  qubits can be used to represent a $2^n$ dimensional Hilbert space. For several of the embeddings we discuss, the computational cost of state preparation is indeed low.   Nevertheless, it has been suggested that for more useful encodings of classical data, the process of state-preparation can suffer from an exponential scaling of the computational costs if the data does not have certain kinds of structure~\cite{aaronson2015read}.   In that case, both the QNN and the TPP would suffer from similar computational inefficiencies. 

\noindent \textbf{What is the promise of alternative strategies for quantum machine learning?}
If finding straightforward quantum advantage for QNNs is not a viable research programme at the moment, perhaps, as suggested in ~\cite{schuld2022quantum}, other goals should be pursued first.  One option is that QNNs could potentially be used to efficiently estimate a classical kernel \citep{schuld2021quantum}. We illustrate this avenue by using a QNN to compute inner products for an infinite-width DNN kernel, finding that this formulation outperforms the standard QNN.  However, this formulation still suffers from standard problems of kernel methods such as poor scaling with data size. For classical data, there remains the problem of encoding the data onto the quantum computer.

\noindent \textbf{What is the promise of alternative strategies for quantum machine learning?}
If finding straightforward quantum advantage for QNNs on classical data is not a viable research programme at the moment, perhaps, as suggested in ~\cite{schuld2022quantum}, other goals should be pursued first.  One option is that QNNs could potentially be used to efficiently estimate a classical kernel. We illustrate this avenue by using a QNN to compute inner products for an infinite-width DNN kernel, finding that this formulation outperforms the standard QNN.  However, this formulation still suffers from standard problems of kernel methods such as poor scaling with data size and weaker performance than state-of-the-art DNNs on large datasets. Furthermore, there remains the problem of encoding the classical data onto the quantum computer.

We briefly consider how to move beyond the QNN/kernel paradigm with layers of PQCs and intermediate measurements to create deep QNNs.  Using the correspondence between the single output QNN and the perceptron we prove universal approximation theorems for these models and experimentally demonstrate that they have better inductive bias for both the Boolean data and Q-FashionMNIST than simpler QNNs do.   However, these models are likely to suffer from computational efficiency problems that scale badly with size. 

\section{Conclusions}
Our results illustrate a growing (but by no means unanimous) consensus in the literature~\citep{preskill2018quantum,schuld2022quantum,kubler2021inductive,huang2021power,gyurik2022establishing,schreiber2023classical,bowles2024better} that current QNN algorithms are unlikely to rival DNNs as general learners on classical data, even with future advances in quantum computing beyond the NISQ regime. Perhaps alternative methods such as DQNNs, or hybrid methods where a QNN provides efficient calculations within a broader classical algorithm may bring some benefits. But, this hope is tempered by the dauntingly wide performance gap between current quantum machine learning algorithms and state-of-the-art classical deep neural networks on large datasets.  

By contrast, QNNs and their extensions may have advantages for learning certain kinds of quantum data~\cite{huang2021power,huang_quantum_2022}. Such advantage is most likely to emerge if the problems exhibit special structure, as is the case for other quantum algorithms~\cite{aaronson2015read,aaronson2022much}.  For example, one requirement for QNNs is linear separability in Hilbert space.

More generally,   using quantum computers as quantum simulators hold exciting promise for near-term advances in fundamental physics problems where such special structures exists~\cite{daley2022practical,mi2024stable}. For instance, the ability to precisely control strongly correlated quantum systems out of equilibrium could generate new insights into complex materials like high-temperature superconductors, potentially yielding commercial impact. 
Searching for quantum problems with the kinds of special structures that QNNs can exploit may lead to an important new scientific understanding into the nature of quantum mechanics, making such research valuable even in the absence of commercial applications.  

This paper is pessimistic about the prospects of QNNs outperforming DNNs on classical data. On the other hand, proving a negative is hard and surprises abound in science.  Moreover,  there is a wider class of QML systems beyond QNNs which may evade some of the problems we . Quantum computing has established exponential advantages in other domains and there may still be undiscovered avenues to harness its potential in quantum machine learning.

\begin{acknowledgments}
C.M. acknowledges funding from an EPSRC iCASE grant with IBM (grant number EP/S513842/1). C.L. acknowledges funding from the EPSRC. J.P. acknowledges funding from the EPSRC Doctoral Training Partnership.
We thank Garnet Chan, Shivaji Sondhi, Nik Gourianov, and Douglas Brown for useful conversations.
\end{acknowledgments}

\FloatBarrier
\bibliography{bib}

\appendix

\newpage
\onecolumngrid
\input{appendix}

\end{document}

%% file: packages.tex
\usepackage{graphicx} % Required for inserting images
\usepackage{physics} % For Braket notation
\usepackage{listings} % code

\usepackage{amsthm}

\usepackage[colorlinks=true,breaklinks=true,linkcolor=magenta,citecolor=magenta,urlcolor=magenta]{hyperref} % Magenta citations and hyperlinks
\usepackage[colorlinks=true]{hyperref} % Color links

\usepackage{placeins} % Used for FloatBarrier which creates a barrier
\usepackage{enumitem} % Used for [noitemsep] which removes the space between items in enumeration
\usepackage{subcaption} % for subfigures
\usepackage{amsfonts} % for mathbb
\usepackage{todo}
\usepackage{titlesec}
\usepackage{bbm}

\usepackage{tikz}
\usetikzlibrary{quantikz}
\usepackage{tikz-cd}
\tikzcdset{scale cd/.style={every label/.append style={scale=#1},
    cells={nodes={scale=#1}}}}
\usepackage{adjustbox}
\usetikzlibrary{arrows}
\tikzset{
operator/.append style={fill=purple!20},
my label/.append style={above right,xshift=0.3cm},
phase label/.append style={label position=above}
}

\newcommand\tikznode[3][]%
   {\tikz[remember picture,baseline=(#2.base)]
      \node[minimum size=0pt,inner sep=0pt,#1](#2){#3};%
   }
\tikzset{>=stealth}

\newcommand*{\gateStyle}[1]{{\textsf{\itshape #1}}}
\newcommand*{\hGate}{\gateStyle{H}}

\newrobustcmd{\B}{\bfseries}

\newcommand*{\circuitH}{\gate[style={fill=teal!20},label style=black]{\textnormal{\hGate{}}}}

\usepackage{mathtools}
\DeclarePairedDelimiter{\ceil}{\lceil}{\rceil}

\makeatletter
\newcommand{\ostar}{\mathbin{\mathpalette\make@circled\ast}}
\newcommand{\make@circled}[2]{%
  \ooalign{$\m@th#1\smallbigcirc{#1}$\cr\hidewidth$\m@th#1#2$\hidewidth\cr}%
}
\newcommand{\smallbigcirc}[1]{%
  \vcenter{\hbox{\scalebox{0.77778}{$\m@th#1\bigcirc$}}}%
}
\makeatother

\newcommand{\tpps}{TPPs~}
\newcommand{\tpp}{TPP~}
\newcommand{\tppn}{TPP}
\newcommand{\dqnna}{DQNN-$\alpha$~}
\newcommand{\dqnnan}{DQNN-$\alpha$~}
\newcommand{\dqnnb}{DQNN-$\beta$~}
\usepackage{hyperref}

%% file: appendix.tex
\section{A brief review of some key papers on kernels and QNN performance }\label{sec:lit_review}

In this section, we describe in a bit more detail some key papers on kernels and QNN performance that we also discuss in the main text. 

\citet{huang2021power} compares the inductive biases of QNN-inspired and classical kernels. They define a geometric measure of similarity between two kernels $K_1$ and $K_2$ that they call $g(K_1\|K_2)$, which is based on the spectral norm of a function of $K_1$ and $K_2$. 
%investigated the inductive biases of QNN-inspired kernels. It provides conditions for quantum advantage (when a QNN will outperform classical learning algorithms) 
%or an ability to learn. The authors use a geometric measure of similarity between two kernels they call $g$ (see the paper for details). 
They argue that if a classical kernel $K_C$ can be computed efficiently such that it is sufficiently well aligned to the target quantum kernel $K_Q$ -- such that $g({K_C \| K_Q})<< \sqrt{N}$ (where $N$ is the number of datapoints) then quantum advantage is not possible. 
%Whether the systems can learn is a question of whether the inductive bias of the kernel is sufficiently good.
Quantum advantage is possible when approximating $\bra{x}\ket{x'}$ classically cannot be done to the required degree of precision -- that is, $g({K_C \| K_Q})\propto \sqrt{N}$ for any efficiently computable $K_C$. Note that even in this case, $K_C$ may still be able to learn the target function (e.g.\ a constant target function would be easily learned). For quantum advantage to occur, $K_Q$ must also be able to learn the target, and $K_C$ must not. By optimising these measures they engineer datasets that they claim are hard to learn with classical ML, but can be more easily learned with a QNN, although the advantage decreases with more data.

\citet{kubler2021inductive} show that if the reproducing kernel Hilbert space (RKHS) of the quantum kernel is low dimensional and contains functions that are hard to compute classically, quantum advantage is possible. If the target function $f$ is well-aligned to the principle components of the kernel $K_Q$, $f$  then it can easily be learned, formalizing the learnability criteria in \citet{huang2021power}.
\citet{kubler2021inductive} then give conditions on whether learning is possible (note that their results apply to any kernel, not just QNN kernels). They show that learning with polynomial amounts of data is only feasible when the mean embedding of most coordinates is close to a pure state. This means that, using $\phi_i$ to denote the embedding onto the $i$'th qubit, $z\mapsto \ket{\phi_i(z)}$ is almost constant. Specifically, if $\Tr{\rho_\mu^2}\leq \delta<1$, where $\rho_\mu$ is the square of the mean density matrix, then there exists some input dimension for which polynomial data will not suffice for learning.
This means that the potential quantum advantage from a Hilbert space of dimension $2^n$ cannot be utilised with polynomial data.

In this context, we note that \citet{shaydulin2022quantumkernel} introduce a `bandwidth' parameter that controls how spread out the data is, and the parameter can be varied to counteract the bound in \cite{kubler2021inductive}, see also~\cite{canatar2022bandwidth}.  \citet{peters2022generalization} explain why quantum kernels can still learn even when overparameterised, using a similar analysis to that applied to more general kernels (e.g.\ \citep{canatar2021spectral}).

Finally, we mention a study by \citet{wright_capacity_2020} who analysed the expressive power of parametrised quantum circuits using memory capacity. They showed that the capacity $C$ of a QNN is bounded by the information $W$ that can be trained into its parameters $C \leq W$. This means that QNNs that are parametrised classically do not have an advantage in capacity over classical neural networks with the same number of parameters, as the bound $C \leq Q$ applies universally, regardless of the learning machine and its features. QNNs parameterised with quantum states, however, could have exponentially larger capacities.

\FloatBarrier
\section{Problems with claims of potential quantum advantage for QNNs in the literature}\label{app:abbas_paper}

The main prior experience with machine learning for several authors on the current paper was with classical DNNs. We started this project as relative newcomers to the field of quantum machine learning. We initially read many papers reporting excellent performance on well-known benchmark datasets that we were familiar with. In many cases, these studies claimed quantum advantage over classical DNN architectures. At first, this seemed quite exciting. 
However, upon further reading, we repeatedly found that these claims of quantum advantage were presented in ways that at best we easily misunderstood, or, at worst,  were simply wrong. Working this out often took significant effort, often involving digging through code in repositories to reconstruct what had actually been done.  Two of the key issues we repeatedly found include: 
\begin{enumerate}
 \item Datasets that are so drastically simplified that they bear little resemblance to the originals from which they took their names (see \cref{app:over-simplified} for more detail.
 \item Comparing QNNs to sub-standard DNNs so that the quantum advantage stems from the unnecessarily poor performance of the classical DNN, not from the good performance of the QNN.
 See \cref{app:IBMpowerpaper} for more detail.
\end{enumerate}

While we were finishing this paper, we were encouraged by the appearance of a preprint (\citet{bowles2024better})  written by industry insiders in the field of QML, that points out similar serious problems with many empirical claims of quantum advantage. While there are also many good papers in the field, the number with serious problems is remarkably high. The situation is so dire that the authors of \cite{bowles2024better} felt the need to directly address the question of whether or not they are dealing with deliberate deception. They write: \textit{Since leaderboard-driven research actively searches for good models, a positivity bias of this nature is not a question of ethical misconduct, but built into the research methodology.} In other words, a culture has developed that makes it intellectually (and perhaps even morally) hard to get certain key things right. They give a road map towards better practices when empirically benchmarking quantum machine learning against classical methods. We hope that this kind of important work will raise the standards of how quantum advantage is analysed and reported.

\subsection{Example of over-simplified data-set}
\label{app:over-simplified}

Fashion-MNIST~\citep{xiao2017fashion} is a popular dataset consisting of 70,000 images of fashion products, split into 60,000 images for the training set and 10,000 images for the test set. Each image is made up of 28x28 grayscale pixels with values 1 to 255. These images are associated with 10 labels corresponding to clothing items such as t-shirts/tops, trousers, pullovers, dresses, bags, etc. 
\iffalse
\begin{description} 
\item[0] T-shirt/top
\item[1] Trouser
\item[2] Pullover
\item[3] Dress
\item[4] Coat
\item[5] Sandal
\item[6] Shirt
\item[7] Sneaker
\item[8] Bag
\item[9] Ankle boot associated with a label from 10 classes.
\end{description}
\fi
To a reader versed in classical ML, its use signals an attempt to try a more difficult classification problem than the famous MNIST dataset of handwritten digits~\citep{lecun1998mnist}. Interestingly, many QNN papers also claim to use Fashion-MNIST to test QNN performance. At first sight, the choice of this more complex dataset looks impressive. Upon further investigation, however, we found that these papers always use a hugely simplified version of the original dataset. The extreme nature of this simplification is not always that easy to spot in the text where authors describe their work. 

\begin{figure}[h]
 \centering
 \begin{subfigure}{0.45\textwidth}
  \includegraphics[width=\textwidth]{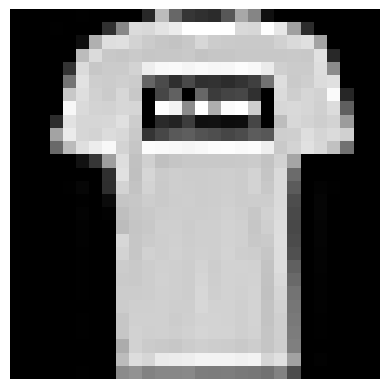}
  \caption{Top, FashionMNIST}
  \label{fig:subfig1}
 \end{subfigure}
 \hfill
 \begin{subfigure}{0.45\textwidth}
  \includegraphics[width=\textwidth]{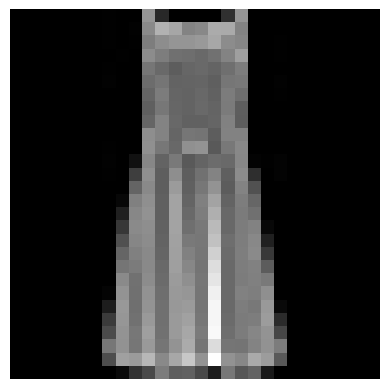}
  \caption{Dress, FashionMNIST}
  \label{fig:subfig2}
 \end{subfigure}
 \\
 \begin{subfigure}{0.45\textwidth}
  \includegraphics[width=\textwidth]{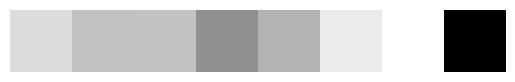}
  \caption{The same top, Q-FashionMNIST}
  \label{fig:subfig3}
 \end{subfigure}
 \hfill
 \begin{subfigure}{0.45\textwidth}
  \includegraphics[width=\textwidth]{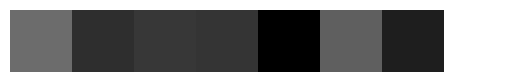}
  \caption{The same dress, Q-FashionMNIST}
  \label{fig:subfig4}
 \end{subfigure}
 \caption{Representative examples of the FashionMNIST and reduced Q-FashionMNIST dataset. Note how strongly the original images are down-sampled from 784 pixels and 10 classes to 8 pixels and 2 classes for Q-FashionMNIST. Very little of the original structure is left over and this dataset can nearly be linearly separated.  }
 \label{fig:mainfig}
\end{figure}

In the main text, we, therefore, introduced the acronym Q-FashionMINST to identify our variant of what is often called FashionMNIST in the QNN literature.  Our construction is very similar to the dataset used by Tensorflow Quantum in \cite{broughton_tensorflow_2021} and in many QNN papers. This dataset is constructed by selecting two classes from FashionMNIST (tops and dresses), and reducing the data dimensionality by selecting the top 8 components via principal component analysis on the entire dataset. We use 250 training datapoints and 50 test datapoints which is typical in the literature, but much smaller than the full FashionMNIST set.   While the new dataset is not entirely linearly separable, a perceptron can achieve  92\% accuracy on Q-FashionMNIST with no need for feature engineering.

As can be seen in \cref{fig:mainfig}, the down-sampled data with just 8 pixels is so much simpler than the original data that any link to the fashion-based classes in  FashionMNIST is at best tenuous. Even our modified name could still be potentially deceptive, given that it continues to suggest a more complex dataset than MNIST. 
More generally, should one downsample a series of image datasets in the same way that Q-FasionMINST is, it is not clear how much of any complexity hierarchy would remain in the downsampled datasets.

There is in principle nothing wrong with using simplified classical datasets to study QNNs. Our complaint here is that the nature of these datasets must be more clearly communicated to the reader (see also the next section where we provide an example for Fisher's Iris dataset).  Many QNN papers fall down at this hurdle. 
This particular example is just one of many problems with the descriptions of datasets that we have found upon reading the QNN literature. 

\subsection{Example of over-simplified classical architectures:}
\label{app:IBMpowerpaper}

We also found many cases where quantum advantage is claimed based on comparisons to DNNs that are enormously simplified, even when to simple off-the-shelf DNNs, let alone state-of-the-art architectures. 

To illustrate what we mean, we will treat in more detail a 2021 paper entitled \textit{``The power of quantum neural networks''} by \citet{abbas2021power}. We chose it because it is influential by the metric of citation count, with over 600 citations on Google Scholar in less than 3 years after publication (checked May 2024). The authors define an effective dimension measure based on the Hessian of the loss function (which in this case reduces to the Fisher information matrix) and use this to predict generalisation performance. They argue that the spectrum of the Hessian of some QNNs is more optimal for learning than some FCNs. Their abstract concludes with the claim that: \textit{We demonstrate numerically that a class of quantum neural networks can achieve a considerably better effective dimension than comparable feedforward networks and train faster, suggesting an advantage for quantum machine learning, which we verify on real quantum hardware.}
We will leave untouched the question of how useful this metric is, given the extensive and sophisticated literature about such measures based on Hessians and Fisher information in the classical machine-learning literature. Instead, we focus on the numerical comparisons of a QNN with classical DNNs.

\begin{figure}[h]
 \centering
 \begin{subfigure}[b]{0.35\textwidth}
 \centering
 \begin{tikzpicture}[node distance=0.75cm]
  % Architecture 1
  \foreach \i in {1,...,4} {
  \node[circle,draw] (x\i) at (0,1.5-\i) {};
  }
  \node[circle,draw] (h1) at (1,-1) {};
  \node[circle,draw] (h2) at (2,-1) {};
  \node[circle,draw] (h3) at (3,-1) {};
  \foreach \i in {1,...,2} {
  \node[circle,draw] (y\i) at (4,0.5-\i) {};
  }

  \foreach \i in {1,...,4} {
  \draw (x\i) -- (h1);
  }
  \draw (h1) -- (h2);
  \draw (h2) -- (h3);
  \foreach \i in {1,...,2} {
  \draw (h3) -- (y\i);
  }
 \end{tikzpicture}
 \caption{Classical FCN from \citet{abbas2021power}.}\label{fig:arch:abbas}
 \begin{tikzpicture}[node distance=0.75cm]
  % Architecture 2
  \foreach \i in {1,...,4} {
  \node[circle,draw] (a\i) at (0,1.5-\i) {};
  }
  \foreach \i in {1,...,2} {
  \node[circle,draw] (b\i) at (2,0.5-\i) {};
  }

  \foreach \i in {1,...,4} {
  \foreach \j in {1,...,2} {
   \draw (a\i) -- (b\j);
  }
  }
 \end{tikzpicture}
 \caption{Our perceptron architecture.}\label{fig:arch:perceptron}
 \label{fig:arch2}
 \end{subfigure}
 \begin{subfigure}[b]{0.6\textwidth}
 \centering
  \includegraphics[width=1\textwidth]{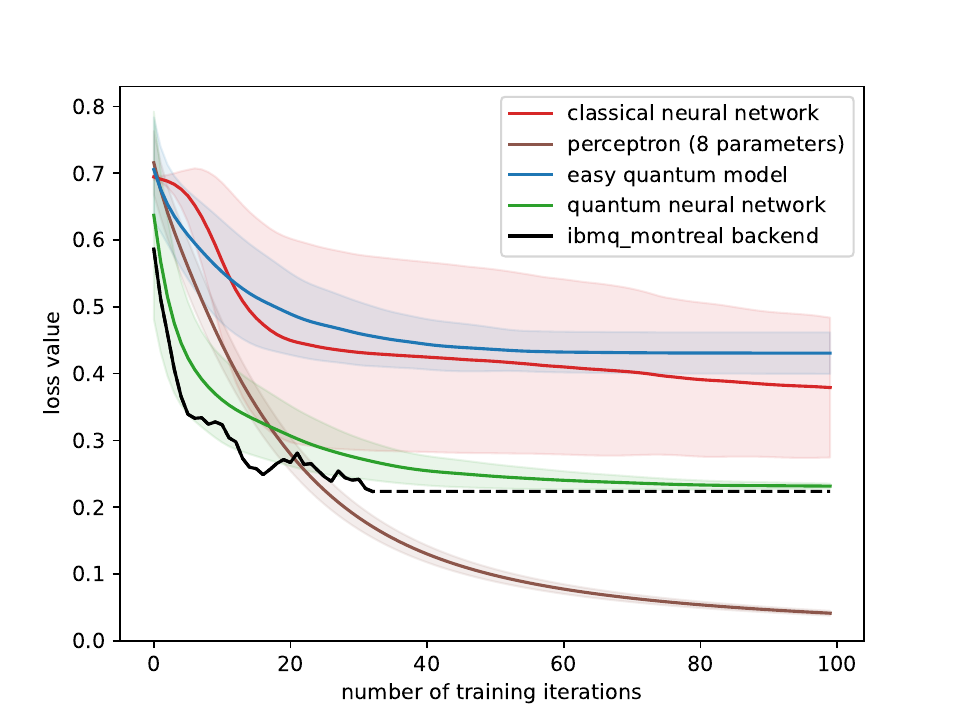}
  \caption{Data}\label{fig:abbas}
 \end{subfigure}
 
 \caption{\textbf{Misleading claims of quantum advantage in \citet{abbas2021power}.} (a) shows a schematic of the FCN architecture used in \citep{abbas2021power} with layer sizes $[4,1,1,1,2]$; it is called the classical neural network in~\cite{abbas2021power}. (b) Our perceptron (layer sizes $[4,2]$). Each architecture has 8 parameters. (c) Equivalent data to Fig 3b from \citet{abbas2021power} shows that their 8-parameter QNN (called quantum neural network)  outperforms their classical neural network (from (a)). However, rather than showing quantum advantage, their FCN's poor performance is caused by the use of single node bottlenecks. We added an 8-parameter perceptron (see (b))  that achieves close to 0 loss, significantly lower than the QNN.  The fact that the perceptron can fully classify this data is not surprising given that the simplified Iris data-set~\cite{fisher1936use} they use is almost trivially linearly separable (see Fig 8.), while the full dataset is not. Note that we have changed our y-axis to start at $0.0$ instead of $0.2$ as done in~\citep{abbas2021power}. 
}
 \label{app:fig:abbas_experiment}
\end{figure}
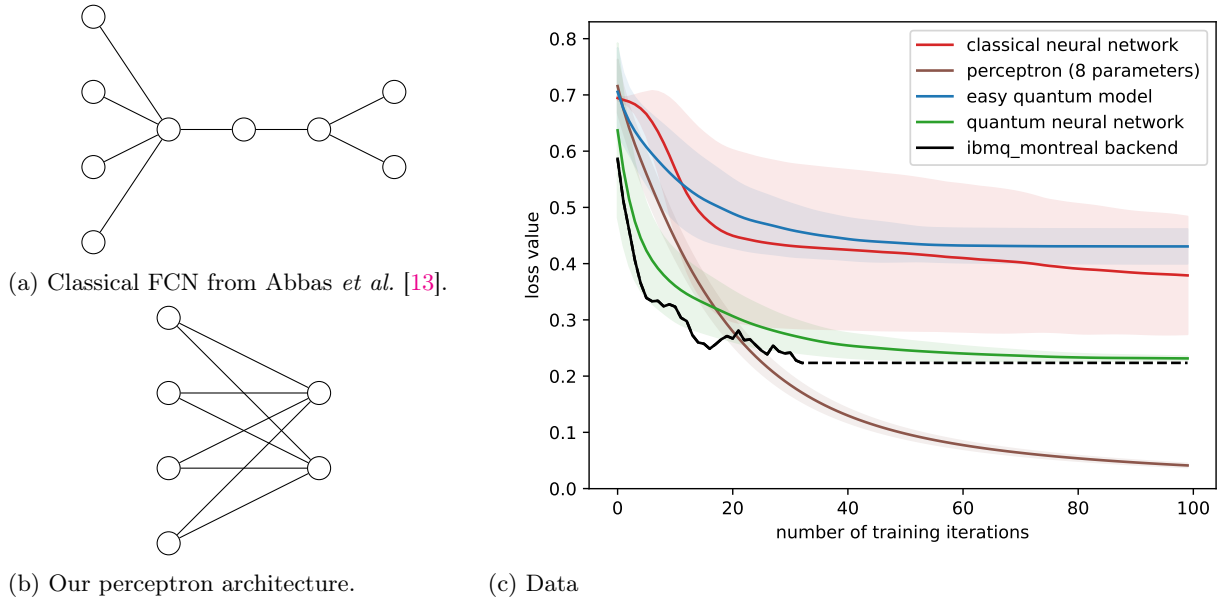

The paper includes numerical experiments where they compare their QNN to a classical feedforward network (FCN) with the same number of parameters on the first two classes of Ronald Fisher's famous 1936 Iris flower dataset \citep{fisher1936use} collected by Edgar Anderson.   In their Fig 3a,  they compare the Fisher-information (the same Fisher who published the Iris flower dataset)  based effective dimension measure, finding that the QNN has significantly higher capacity than the comparable FCN, one of their key results.  
In Fig 3b they provide another set of key results, namely that their QNN trains faster and to significantly lower loss than a comparable classical DNN, implying a potential quantum advantage.

The authors don't provide the actual DNN architecture they use, but do say that: \textit{We consider all possible topologies with full connectivity for a fixed number of trainable parameters. Networks with and without biases and different activation functions are explored.}  At first sight, it does seem impressive that their QNN has higher capacity and performs so much better than what they claim is the best FCN of the same size.

We could not find the DNN described with enough specificity to understand these surprising results.  However, by searching through their codebase\footnote{\url{https://doi.org/10.5281/zenodo.4732830}}, we found the architecture used in their main Fig 3b. They used leaky ReLU activations and layer sizes $[4,1,1,1,2]$ with 8 parameters, the same number as the QNN. We depict it schematically in \cref{fig:arch:abbas} and it is immediately clear that this is an eccentric and sub-optimal architecture with one-node intermediate layers that act as extreme bottlenecks.  Not surprisingly, a very simple 8-parameter perceptron with layer sizes $[4,2]$ (see \cref{fig:arch:perceptron}) already does much better than both their classical DNN and their QNN, as can be seen in \cref{fig:abbas}. 

Another puzzling aspect of this paper concerns the choice of dataset. The Iris flower dataset \citep{fisher1936use} is famous in the history of statistical inference for its pioneering use of features. 
We verified by checking in the code that what they call the first two classes are equivalent to the red and green points in  \cref{app:fig:abbas_experiment_data}.  These are trivially separable making the learning task on a QNN or a DNN more or less meaningless.  This is an extreme example of our problem 1 above, not clearly signalling the nature of an oversimplified dataset to the reader.

\begin{figure}[h]

 \centering
  \includegraphics[width=0.65\textwidth]{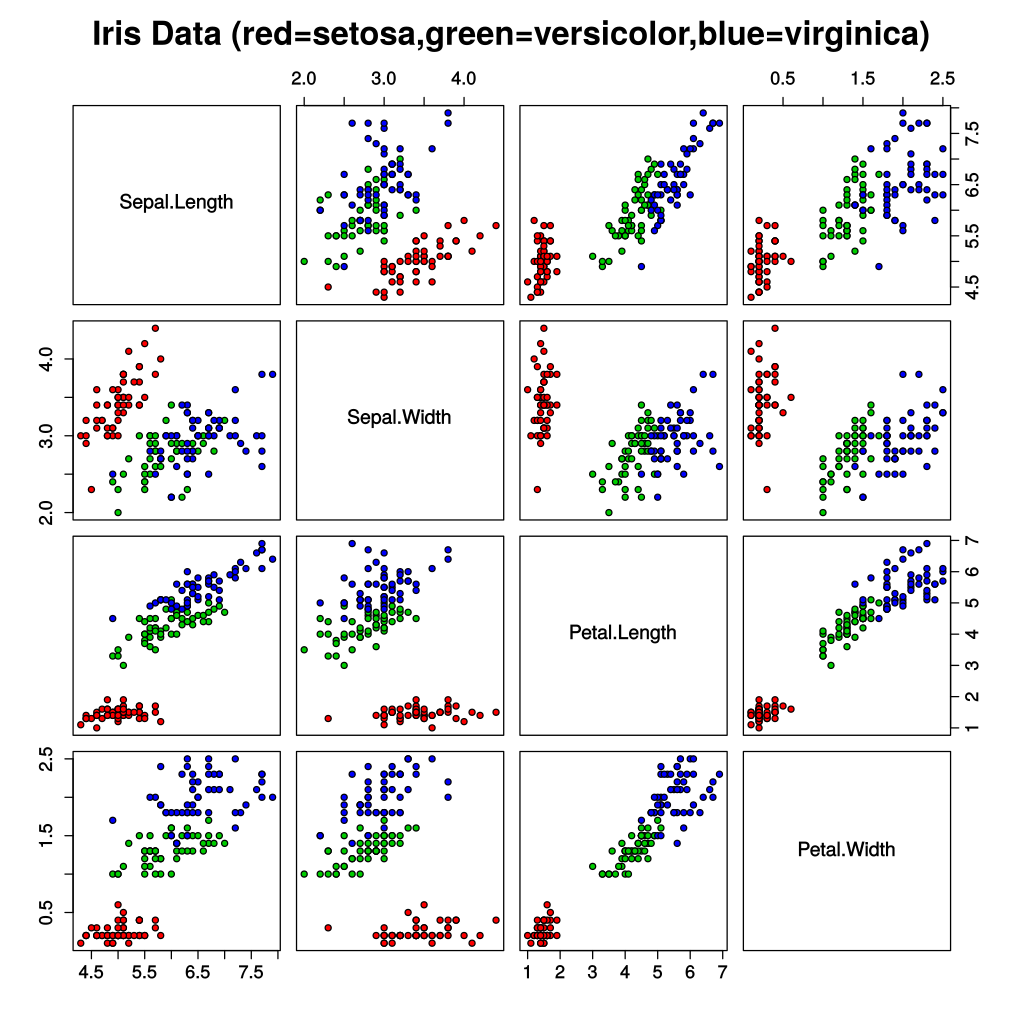}
 \caption{\textbf{The dataset used  in Fig 3 of \citet{abbas2021power} is trivially separable.}  The extreme simplicity of the learning task  can be inferred from this image  of  Fisher's famous Iris dataset~\cite{fisher1936use} (made by Wikipedia contributor Nicoguaro, used under CC BY 4.0 from the Wikipedia page~\cite{irisflowers}).  The authors of \citep{abbas2021power} only used the first two classes (setosa and versicolor) in the experiments shown in  \cref{app:fig:abbas_experiment}, and  should have clearly signalled the unusual nature of this learning task.  
}
 \label{app:fig:abbas_experiment_data}
\end{figure}

This paper and its persistent popularity present some puzzles. Given the simplified dataset, which is very well known (e.g.\ an image can be found on Wikipedia~\cite{irisflowers}), it should be obvious that even a tiny classical DNN should do much better than the DNN performance they presented. Had the authors actually described their classical network, even to themselves, it should also have been immediately obvious that their QNN did better than their DNN for trivial reasons.  
There are further problems with the results in this paper that we will not bother to go into here. This particular example may be fairly egregious,  but it is by no means alone in presenting enormous gaps between rhetoric and reality.    Moreover, the perplexing fact that it continues to be so popular helps illustrate the importance of fixing some of the broader cultural issues critiqued in~\citet{bowles2024better}, and in this Appendix.

Finally, we note that even when papers in the subfield of QNNs are technically correct, there is a tendency to employ overly enthusiastic rhetoric regarding potential quantum advantage which can obscure the clarity of key results. While the prospect of merging quantum computing with artificial intelligence is undoubtedly thrilling,  a shift back to the more measured and occasionally mundane practices of traditional scientific research may in some cases be needed.    

\FloatBarrier

\section{Definitions}
\subsection{The Boolean system}\label{sec:Boolean_system}

\begin{definition}
 The n-dimensional Boolean dataset is completely described by a function
 \begin{equation*}
  f:\{0,1\}^n\rightarrow \{0,1\}.
 \end{equation*}
 We can construct different embeddings of $\{0,1\}^n$ using a map 
 \begin{align*}
  \Phi:\{0,1\}^n &\mapsto X,\\
    i&\mapsto x_i
 \end{align*}
 which maps each binary number $i$ to some input space $x_i\in X$.
\end{definition}

\begin{definition}[The standard representation of $\{0,1\}^n$]\label{def:enc:classical}
 The most straightforward encoding $\Phi$ for $\{0,1\}^n$ is to encode it in $\mathbb{R}^n$ as a binary vector.
 For example, $\phi(0010)=x_{0010}=[0,0,1,0]$. This encodes the data on the Boolean hypercube with one corner at the origin.
 We will use this encoding so often that we will use $\Phi(\{0,1\}^n)$ and $\{0,1\}^n$ interchangeably.

 We will also consider a normalised version of $\{0,1\}^n$, where each element $\norm{x}=1$, such that $\phi(1010)=[1,0,1,0]/\sqrt{2}$. The origin cannot be normalised and is excluded.

 We also consider $\{-1,1\}^n$ -- the centered Boolean hypercube with vectors of $+1$ and $-1$, and its normalised version.
\end{definition}

\begin{definition}[The string representation of $\{0,1\}^n$]
 For some $\{0,1\}^n(f)$
 \begin{equation*}
  f:\{0,1\}^n\rightarrow \{0,1\}.
 \end{equation*}
 $f$ can be represented as a binary string of length $2^n$, by ordering inputs ascending by their binary value, and concatenating their outputs $y$.
 For example, when $n=2$, for the dataset $\{(00,0),(01,1),(10,1),(11,0)\}$, the the string representation would be $f=0110$.
\end{definition}

\textbf{The LZ Complexity} $K(f)$ of $f$ is the complexity of its  binary string representation. Here estimate this complexity using a variation on the famous Lempel-Ziv (LZ) algorithm introduced in 1976 by Lempel and Ziv~\cite{lempel1976complexity} 
to calculate the complexity of binary strings. This compression-based measure has been a popular approximator for Kolmogorov complexity and it is thought to work better than other lossless compression-based measures for shorter strings   For more details on its properties, see \citep{dingle_inputoutput_2018,valle2018deep}. The representation above has the ordering property built in. 
The complexity is therefore $K(f_n)\leq K_{ordering}+LZ(f)$, where $K_{ordering}$ is the complexity of the process ordering the outputs. Using a much more complex ordering would change the values of LZ complexity measured and fundamentally not capture the complexity of the dataset, as $K_{ordering}$ would be much larger.

\textbf{Training set $S_m$.} A training set with $m$ elements, for a dataset $g:\{0,1\}^n\mapsto \{0,1\}$, $S_m = \{(x_i,g(x_i))\}_{i\in s(g)_m}$ where $s(g)_m$ is an index set of size $m$. Normally, the $g$ is left implicit, so $S(g)_m$ is denoted $S_m$.
There are ${2^n \choose m}$ unique training sets of size $m$ for each function $g$. The test set is the complement of the training set.

\subsection{Quantum Encoding Methods}
\label{sec:encoding_methods}
There are multiple ways to encode data into a quantum circuit. Here are the four we use in our work.

\begin{definition}[Basis encoding]\label{def:enc:basis}
This maps the input data (in the form of a binary string) into the computational basis of a quantum state. For vector $x=[b_1, \dots b_n]$ where $b_i\in \{0, 1\}$, basis encoding $\phi_B$ encodes the data as follows
$$
\phi_B(x)=\ket{b_{n-1} b_{n-2} \cdots b_0}
$$
For example, basis encoding transforms the input data $x = (1,0)$ into $\ket{10}$. This requires $O(n)$ qubits. 
\end{definition}

\begin{definition}[Amplitude encoding]\label{def:enc:amp}
 This encodes the data into the amplitudes of the quantum state, such that amplitude encoding $\phi_A$ encodes $x=[x_1, \dots, x_N]$ via
 $$
 \phi_A(x)=\frac{1}{\norm{x}} \sum_{i=1}^{N} x_i \ket{i}
 $$
 where $\ket{i}$ is the $i^{th}$ computational basis state. For example, if our vector $x = (1,0,1,0)$ then the norm is $\sqrt{1^2 + 0^2 + (1)^2+ (0)^2} = \sqrt{2}$, so our encoded quantum state becomes $\frac{1}{\sqrt{2}} \ket{00} + \frac{1}{\sqrt{2}}\ket{10}$. Extra computational basis states can be encoded as zero. This requires $O(\log(n))$ qubits. \\
 \textbf{01 encoding} For the $n$-dimensional Boolean dataset, each vector $x\in[0,1]^n$ (followed by correct normalisation) -- a corner of the Boolean hypercube. Unless otherwise specified, we will refer to this as amplitude encoding. Note that there is no way for $x=[0,\dots,0]$ to be encoded, so that datapoint is dropped.\\
 \textbf{+1-1 encoding} For the $n$-dimensional Boolean dataset, each vector $x\in[-1,1]^n$ (followed by correct normalisation) -- a corner of the centered Boolean hypercube. This will always be referred to as +1-1 encoding, never just as amplitude encoding.
\end{definition}

\begin{definition}[ZZ encoding]\label{def:enc:zz}
 The ZZ Feature Map is based on \cite{havlicek_supervised_2019} and encodes the input data into the angles of the unitary matrices.
 $$
 \phi_{ZZ}(x) = U_{\Phi(x)} H^{\otimes n} U_{\Phi(x)} H^{\otimes n}
 $$
 where $U_{\phi(x)}=\exp \left(i \sum_{S \subseteq[n]} \phi_S(x) \prod_{i \in S} Z_i\right)$. H is the Hadamard gate, Z is the Pauli-Z gate, $S$ is the dataset, and $n$ is the number of qubits. $\phi_i(x)=x_i, \phi_{\{i, j\}}(x)=\left(\pi-x_0\right)\left(\pi-x_1\right)$.
 
 This feature map can be seen more clearly in its circuit diagram for two qubits in Figure \ref{fig:zz_feature_map}.
 
 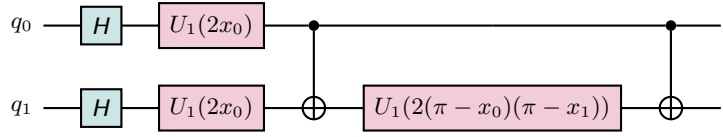
\begin{figure}[H]
  \centering
  \adjustbox{scale=1}{%
   \begin{tikzcd}
   \lstick{$q_0$} & \circuitH & \gate{U_1(2x_0)} & \ctrl{1} & \qw & \ctrl{1} & \qw \\
   \lstick{$q_1$} & \circuitH & \gate{U_1(2x_0)} & \targ{} & \gate{U_1(2(\pi-x_0)(\pi-x_1))} & \targ{} & \qw \\
  \end{tikzcd}
  }
 \caption{ZZ Feature Map for two qubits}
 \label{fig:zz_feature_map}
 \end{figure}
\end{definition}

\begin{definition}[Random Transform (RT) encoding]\label{def:enc:relu}
This encodes the data into the amplitudes of the quantum state, after transforming through a single hidden layer of a neural network. Let $x=[x_1, \dots, x_n]$
$$
\phi_{R}(x) = \frac{1}{\norm{\sigma(Wx+b)}} \sum_{i=1}^{N}\sigma(Wx+b)_i \ket{i}
$$
where $\sigma(Wx+b)_i$ is the $i$'th output of a single hidden layer neural network with input dimension $n$ and output dimension $N$, weights $W$, biases $b$ and non-linearity $\sigma$. In our experiments, the RT encoding was created by a fixed random transform (parameters in $W$ and $b$ sampled i.i.d.) and using the ReLU nonlinearity.

\end{definition}

\subsection{Boolean target function details \& training the \tpp}\label{app:targ_train_details}

The 100 functions are generated as described below
\begin{enumerate}
 \item The parity function
 \item We then chose functions with a fixed number $t$ of 1s (with elements randomly shuffled), with $t=0,4,8,\dots, 128$ (33 in total). These generate random functions with fixed class balance.
 \item We picked 10 functions with $p$-fold symmetry where $p=2,4,8,16,\dots$, removing duplicates (54 total). This is done by generating a string of length $2^n/p$ and repeating it $p$ times.
 \item 12 random uniformly chosen functions to make it up to 100
\end{enumerate}
See the Supplementary Information for the exact functions.

Training the \tpp version of the QNN requires first transforming $\ket{x}$ into the generalised form of $x\ostar x$ as described in \cref{app:perceptron_proof}. The \tpp is then trained on $x\ostar x$ with mse loss and SGD with batch size half of the training set size. At the end of training, weights can be scaled to $w/\norm{a(w)}_2$ (to guarantee an equivalent QNN model), as this does not affect the classification.

\section{Custom parity embedding}\label{app:zz}

A QNN acting on the Boolean data with ZZ encoding is biased towards the parity function.
It is worth demonstrating how an FCN can be biased just as strongly towards parity. Consider a map from an $n$ dimensional Boolean vector $x$ to a one-hot vector $x'$ of length $n+1$, where the 1 is in the index equalling the sum of $x$ (where the first element denotes a sum of $0$, so for $x=[0,1,1]$, $x'=[0,0,1,0]$ and $x=[0,0,0]$, $x'=[1,0,0,0]$). Concatenate $x$ with $x'$, and call the resulting encoding for $x$ the parity encoding. We train a standard 1 hidden layer FCN with ReLU activations (using the parity function as the target) on this parity encoding and compare its performance to the standard encoding in \cref{fig:parity_encoding}.

\begin{figure}[H]
 \centering
 \includegraphics[width=0.35\textwidth]{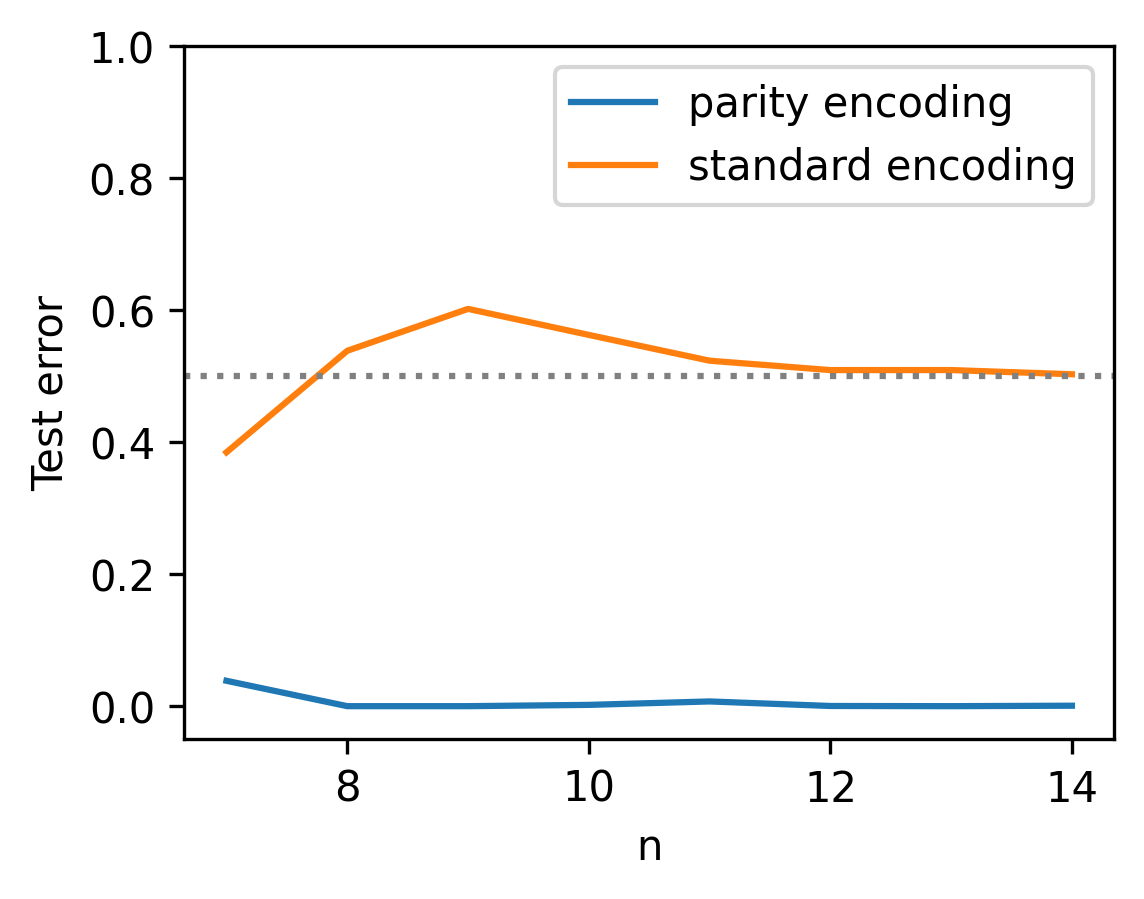}
 \caption{Test error of a 1 hidden layer FCN on the parity function. The FCN is trained with batch size 32 and randomly selected training data 
 with $\min(0.75\times 2^n, 512)$ examples.
 The parity encoding (\cref{app:zz}) allows the classical neural network to be strongly biased towards parity, whereas the standard encoding performs poorly.}
 \label{fig:parity_encoding}
\end{figure}

\section{Proof of equivalence of a QNN to a perceptron and its consequences for expressivity on Boolean data}\label{app:perceptron_proof}

\begin{lemma}[Parameter counting in QNNs]\label{app:lemma:num_parameters}
 A QNN acting on an input Hilbert space $\mathbb{C}^{2^n}$ has ${(2^n)}^2$ independent parameters
\end{lemma}
\begin{proof}
 Writing out the action of the QNN explicitly,
 \begin{align}\label{eq:explicit_QNN(f)}
  \ket{0}\ket{x}=(x, 0^{(2^n)})^T\quad
  Z = 
  \begin{bmatrix}
   I^{(2^n)} & 0 \\
   0 & -I^{(2^n)}
  \end{bmatrix}\quad
  U = 
  \begin{bmatrix}
   a & b \\
   c & d
  \end{bmatrix}.
 \end{align}
 where $a,b,c,d$ are all matrices in $\mathbb{C}^{2^n\times 2^n}$.

 We can therefore simplify
 \begin{align}
  x^\dag U^\dag Z U x = x^\dag (a^\dag a - c^\dag c) x = x^\dag (2a^\dag a - I) x,
 \end{align}
 using the fact that $a^\dag a + c^\dag c = I^{N}$ as $U$ is unitary.
 Given that $2a^\dag a - I$ is Hermitian, the model has ${(2^n)}^2$ independent parameters.
\end{proof}

\begin{lemma}\label{lemma:arbitrary_matrix}
 Any $k\times k$ matrix $M$ can be embedded into a $2k\times 2k$ unitary matrix $U$ if and only if the singular values of $M$ are upper bounded by $1$.
\end{lemma}
\begin{proof}
 Perform singular value decomposition on $M=RDV$ where $R,V$ are unitary and $V$ diagonal.
 $$
 U=\begin{bmatrix}
  M & R\sqrt{1-D^2}V\\
  R\sqrt{1-D^2}V & -M\\
 \end{bmatrix}
 $$
 The square roots are taken elementwise, and only return real numbers when the singular values are upper bounded by 1. 
 To show the converse, imagine $M$ has singular value $\lambda>1$ and corresponding normalised vector $\ket{\lambda}$ 
 \begin{align*}
  U=
  \begin{bmatrix}
  M & A\\
  B & C\\
  \end{bmatrix} \quad\quad\quad\quad
 U\begin{bmatrix}
  \ket{\lambda} \\
  0 \\
  \end{bmatrix}=
  \begin{bmatrix}
  M\ket{\lambda} \\
  B\ket{\lambda} \\
  \end{bmatrix}
 \end{align*}
 The output of this state must have a norm greater than the norm of $M\ket{\lambda}$, which if $\lambda>1$ must be greater than 1.\footnote{\texttt{https://quantumcomputing.stackexchange.com/questions/5167/when-can-a-matrix-be-extended-into-a-unitary/}}.
\end{proof}

\begin{restatable}{lemma}{MainThm1} \label{proof:tpp=qnn}
    \textnormal{(Perceptrons on $x\ostar x$ can express any pre-thresholded function a QNN can)}
    Consider a QNN acting on inputs $x\in \mathbb{C}^{2^n}$, expressing the function $f(x;\theta)=x^\dag U(\theta)^\dag Z U(\theta) x$, with parameters $\theta$. Consider a perceptron with weights $w\in R^{(2^n)^2}$ expressing the function $g(x;w)=w\cdot (x\ostar x)$, where
 \begin{align}
 (x \ostar x)_{N(i-1)+i} &= x_i x_i^* \\
 (x \ostar x)_{N(i-1)+j} &= Re(x_i x_j^*) + Im(x_i x_j^*) \;\;\; \text{for } i \neq j.
 \end{align}
 For any set of parameters $\theta$, there exists some weight $w$ where $f(x;\theta)=g(x;w)$. Note that QNNs have a further restriction on input norm, $\norm{x}=1$.
 We say pre-thresholded function to distinguish the raw output of the QNN or \tpp, before thresholding, from the classification (which we normally just call function).
\end{restatable}

\begin{proof}
 From \cref{eq:explicit_QNN(f)}, we have $f(x)=x^\dag U(\theta)^\dag Z U(\theta) x=x^\dag A x$, where 
 \begin{equation}\label{eqn:A}
  A=A^\dag=(2aa^\dag-I).
 \end{equation}
 $a\in C^{2^n \times 2^n}$ is the upper left quarter of $U$, defined in \cref{eq:explicit_QNN(f)}, and as a result, a function of $\theta$. If we consider a 2D input Hilbert space, $x=(x_1, x_2, 0, 0)$, let $x_i=c_i + i d_i$,
 $$
 f(x) = A_{11}(c_1^2+d_1^2)+A_{22}(c_2^2+d_2^2)+A'(c_1 c_2+d_1 d_2)+A''(d_1c_2-d_2 c_1)
 $$
 where $A'=A_{12}+A_{21}$ and $A''=i(A_{12}-A_{21})$ are both real coefficients, as $A^\dagger = A$.
 $f(x)$ can therefore be replaced by a perceptron of the form $f(x) = w \cdot (x \ostar x) $ with weights 
 $$w=[A_{11}, A'_{12}, A'_{21},A_{22}],$$
 where $A'_{12}=1/2(A'+A'')=Re(A_{12})+Im(A_{12})$ and $A'_{21}=1/2(A'-A'')=Re(A_{12})-Im(A_{12})$,
 acting on
 $$x \ostar x=[c_1^2+d_1^2, c_1 c_2+d_1 d_2+d_1c_2-d_2 c_1, c_1 c_2+d_1 d_2-(d_1c_2-d_2 c_1),c_2^2+d_2^2] = x \ostar x.$$ 
 When $x$ is real everywhere, $c_1=c_2=0$, so $x_1=d_1$, $x_2=d_2$, so $h$ reduces to
 $$
 x \ostar x=[x_1^2, x_1 x_2, x_2 x_1, x_2^2]=x\otimes x.
 $$
 
 This result can straightforwardly be extended to $N$ dimensions. In that case, the weights of the perceptron would be
 \begin{equation}
  \begin{aligned}
  w_{N(i-1) + i} &= A_{ii} \\
  w_{N(i-1) + j} &= Re(A_{ij}) - Im(A_{ij}) \;\;\; \text{for } i \neq j 
 \end{aligned}
 \label{eqn:qnn-percep-weights}
 \end{equation}
 
 Therefore, any QNN with parameters $\theta$ can be expressed by a perceptron on $x \ostar x$, with weights $w$ given by the equations above. 
\end{proof}

In practice, one can use \cref{eqn:qnn-percep-weights} to map any QNN onto its corresponding TPP. 

\begin{restatable}{lemma}{MainThm2}\label{proof:qnn=tpp}
\textnormal{(A QNN can express any (post-thresholded) function a \tpp can)}
 Consider a \tpp $g(x;w)$ and QNN $f(x;\theta)$. Then, given some weights $w$ we can always find a scalar $a$ such that $g(x;w) = a f(x;\theta)$. This means a QNN can express any post-thresholded function (i.e.\ classification) a \tpp can.
\end{restatable}

\begin{proof}

We can use \cref{eqn:qnn-percep-weights} to build the matrix $A$ (defined in \cref{eqn:A}) from the parameters $w$. Given some $A$, $\frac12 (A+I)$ is Hermitian, and so can always be written as $a^\dag a$ for some $a$.
By \cref{lemma:arbitrary_matrix}, we can embed this $a$ in a unitary matrix $U$ provided that the singular values of $a$ are all less than or equal to $1$. If this is possible, then the QNN with unitary $U$ expresses the same pre- and post-thresholded function $g(x)$.

If this condition is not satisfied, simply dividing by the spectral norm of $a$, $w\rightarrow w/\norm{a(w)}_2$, satisfies the above spectral condition, and allows the new $w$ to be expressed by a QNN. 

As classifications are not affected by the scale of $w$, a QNN can classify any data if and only if the \tpp $g(x;w)$ can.

\end{proof}

\begin{restatable}{lemma}{qnnkernel}\label{lemma:qnnkernel}
The \tpp $g(x;w)=w\cdot (x \ostar x)$ has the same kernel as the QNN, given by $K(x^{(1)}, x^{(2)})=\norm{x^\dag x'}^2$.
\end{restatable} 
\begin{proof}
A perceptron $f(h)=w\cdot h$ has a kernel $K(h,h')=h\cdot h'$. Using $S=Re(x_i x^*_j)$ and $T=Im(x_i x_j^*))$, and letting $h=x\ostar x$,
\begin{align*}
(x\ostar x)\cdot (x'\ostar x')' &= (x_i {x_i}^*)(x'_i {x'_i}^*)
+ (S+T)(S'+T')
+ (S-T)(S'-T')
+ (x_j {x_j}^*)(x'_j {x'_j}^*)+\dots\\
&=(x_i {x_i}^*)(x'_i {x'_i}^*)
+ 2(SS'+TT')
+ (x_j {x_j}^*)(x'_j {x'_j}^*)+\dots\\
&=(x_i {x_i}^*)(x'_i {x'_i}^*)
+ Re(x_i {x^*}_j)Re(x'_i {x'^*}_j)
- Im(x_i {x_j}^*)Im(x_i {x_j}^*)
+ (x_j {x_j}^*)(x'_j {x'_j}^*)+\dots\\
&= \norm{x^\dag x'}^2.
\end{align*}
\end{proof}

\section{QNNs with amplitude encoding cannot express XOR for $n>2$}\label{sec:sec:qnn_amplitude_parity_proof}

\begin{lemma}\label{app:lemma:n=2}
 The \tpp can express the XOR function on the $n=2$ Boolean dataset encoded as the Boolean hypercube $\{0,1\}^2$ or the centered hypercube $\{-1,1\}^2$ (see \cref{def:enc:classical}).
 XOR can also be expressed when the inputs $x$ are normalised ($\norm{x}_2=1$, leaving out the origin). The normalised inputs are those received by the QNN.
\end{lemma}

\begin{proof}
 The XOR function is given by $f(x_{00})=f(x_{11})=1$ and $f(x_{01})=f(x_{10})=0$.

 \textbf{Unnormalised $\mathbf{\{0,1\}^n}$} For $\{0,1\}^2$, the embeddings for the \tpp are ${(x\otimes x)}_{00}=[0,0,0,0]$, ${(x\otimes x)}_{10}=[1,0,0,0]$, ${(x\otimes x)}_{01}=[0,0,0,1]$ and ${(x\otimes x)}_{11}=[1,1,1,1]$. A \tpp with weights $w=[1,-1,-1,1]$ correctly expresses XOR. 
 
 \textbf{Normalised $\mathbf{\{0,1\}^n}$} When the inputs are normalised (discarding $x_{00}$), a \tpp with the same weights $w$ also expresses XOR.

 \textbf{Unnormalised $\mathbf{\{-1,1\}^n}$} For $\{-1,1\}^2$, the embeddings for the \tpp are ${(x\otimes x)}_{00}=[1,1,1,1]$, ${(x\otimes x)}_{10}=[1,-1,-1,1]$, ${(x\otimes x)}_{01}=[1,-1,-1,1]$ and ${(x\otimes x)}_{11}=[1,1,1,1]$. In this case, a \tpp with weights $w'=[1,1,1,1]$ expresses XOR.
 
 \textbf{Normalised $\mathbf{\{0,1\}^n}$} A \tpp with weights $w'$ also expresses XOR.

 These results imply that the QNN can express XOR with $n=2$ on both $+1-1$ and $01$ amplitude encoding, by using the correspondence between the \tpp and QNN in \cref{proof:qnn=tpp}.
\end{proof}

\begin{lemma}\label{app:lemma:n=3:01}
 The \tpp cannot express the parity function on the $n\geq 3$ Boolean hypercube $\{0,1\}^n$ (see \cref{def:enc:classical}). This is also impossible when the inputs $x$ are normalised ($\norm{x}_2=1$, leaving out the origin). For completeness, we prove this when we allow the \tpp a bias term, i.e.\ $g(x) = w\cdot (x\ostar x)-z$, where $z\in\mathbb{R}$.
\end{lemma}
\begin{proof}

\textbf{Normalised $\mathbf{\{0,1\}^3}$.} We first prove this for $x\in \{0,1\}^3$, when $x$ is normalised (so e.g.\ $x_{101}=[1,0,1]/\sqrt{2}$). \cref{app:eqn:n=3-01} shows the equations for all inputs (temporarily leaving in the origin, even though it does not belong in this dataset as it cannot be normalised). The targets are on the left, and the weights are written as $w=[,\dots, i]$

\begin{equation}\label{app:eqn:n=3-01}
\begin{bmatrix}
0\\
1\\
1\\
0\\
1\\
0\\
0\\
1\\
\end{bmatrix}=
\begin{bmatrix}
0 & 0 & 0 & 0 & 0 & 0 & 0 & 0 & 0\\
0 & 0 & 0 & 0 & 0 & 0 & 0 & 0 & 1\\
0 & 0 & 0 & 0 & 1 & 0 & 0 & 0 & 0\\
0 & 0 & 0 & 0 & 1/2 & 1/2 & 0 & 1/2 & 1/2\\
1 & 0 & 0 & 0 & 0 & 0 & 0 & 0 & 0\\
1/2 & 0 & 1/2 & 0 & 0 & 0 & 1/2 & 0 & 1/2\\
1/2 & 1/2 & 0 & 1/2 & 1/2 & 0 & 0 & 0 & 0\\
1/3 & 1/3 & 1/3 & 1/3 & 1/3 & 1/3 & 1/3 & 1/3 & 1/3\\
\end{bmatrix}
\begin{bmatrix}
a\\
b\\
c\\
d\\
e\\
f\\
g\\
h\\
i\\
\end{bmatrix}
-z
=
\begin{bmatrix}
0<z\\
i\geq z\\
e\geq z\\
e+f+h+i<2z\\
a\geq z\\
a+c+g+i<2z\\
a+b+d+e<2z\\
a+b+c+d+e+f+g+h+i\geq 3z\\
\end{bmatrix}
\end{equation}
Without using the first inequality ($0<z$) as $000$ cannot be encoded for the QNN, we can prove that the parity function cannot be expressed. If we add the three inequalities with 4 terms on the LHS, and add the resulting inequality to the final inequality, we have
$
a+e+i<3z
$
However, from the first three inequalities, we have 
$
a+e+i\geq 3z.
$
These conditions cannot be satisfied simultaneously.

\textbf{Unnormalised $\mathbf{\{0,1\}^3}$} If we drop the normalisation condition (meaning we can include the origin) we still cannot express XOR.
For unnormalised inputs (i.e.\ the same matrix in \cref{app:eqn:n=3-01} but all positive coefficients are $1$) can be easily proved. We have $0<z$ from the first inequality, and with $a-z, i-z, e-z\geq 0$, we have $0\geq a+e+i+b+c+d+f+g+h-z+(a-z+e-z+i-z)=(a+i+c+g-z)+(e+i+h+f-z)+(a+e+b+d-z)+(-z)$. Each of the bracketed terms is less than 0, but their sum must be positive, creating a contradiction.

\textbf{Extending to $\mathbf{n}$} 
The general input to the \tpp is $x\otimes x$ for $x\in\{0,1\}^n$. 
We split the components of $x \otimes x$ into two groups: $x_ix_j$ where both $i,j\leq 3$, and a second group containing all others. 
Consider the $7$ inputs where all terms in the second group are 0 (which can be achieved by setting $x_i=0$ for $i>3$). These inputs must satisfy the same simultaneous equations as in \cref{app:eqn:n=3-01}, and thus the parity function cannot be expressed.
\end{proof}

\begin{lemma}\label{app:lemma:n=3:+1-1}
 The \tpp cannot express the parity function on the $n\geq 3$ centered Boolean hypercube $\{-1,1\}^n$ (see \cref{def:enc:classical}). This is also impossible when the inputs $x$ are normalised ($\norm{x}_2=1$, leaving out the origin). For completeness, we prove this when we allow the \tpp a bias term, i.e.\ $g(x) = w\cdot (x\ostar x)-z$ where $z\in\mathbb{R}.$
 \end{lemma}
\begin{proof}
For $x\in\{-1,1\}$, the normalisation factor is the same for every input, so we only need to consider the unnormalised case. 

\textbf{Normalised and unnormalised $\mathbf{n=3}$} Note that $g(x)=g(-x)$. For $n=3$ (and all odd $n$), the target for $x_{000}=(-1, -1, -1)$ is $0$, and the target for $x_{111}=(1,1,1)$ is $1$. However, $x_{000}=-x_{111}$, so $g(x_{000})=g(x_{111})$, meaning the \tpp cannot express parity when $x\in\{-1,1\}$.

\textbf{Extending to $\mathbf{n}$} Now consider $n$ and $n+1$, with $n$ odd. Any input can be written as 
$x_{n+1}=[x_{n}, \pm 1]$.
When taking the tensor product, we get
\begin{align}
 S=(x, 1)\otimes (x, 1) &= [x\otimes x, x, x,1]\\
 T=(x, -1)\otimes (x,-1)&= [x\otimes x, -x, -x, 1]\\
\end{align}
S and T must be of different classes, as they have one sign different. Without loss of generality, let $S>0$ and $B<0$. Then, consider a vector of parameters $[c_{\otimes}, c, c_1]$. Then, $c\cdot x >0$ whenever $S(x)>0$. Consider $x=[1,\dots, 1]$, all permutations of $x=[-1,-1,1\dots,1]$ all the way to $x=[-1,\dots,-1,1]$. These are all of the same class. Then, straightforwardly $c_i>0$ at every element. But, to satisfy this for every $x$, $c_i=0$. If this is the case though, we are back to classifying the odd $d$ case, which is impossible, and thus a contradiction.
\end{proof}

\begin{restatable}{theorem}{ExpressThm}\label{app:thm:n>3}
    \textnormal{(Parity inexpressible for a \tpp with $n\geq 3$)}. A \tpp, acting on Boolean data $\{0,1\}^n$ or $\{-1,1\}^n$, for both normalised and unnormalised $x$ cannot express the parity function for $n\geq 3$. As a result, QNNs cannot express the parity function on amplitude-encoded inputs, with either 01 encoding or +1-1 encoding.
\end{restatable}

\begin{proof}
 \cref{app:lemma:n=2} proves that XOR can be expressed for $n=2$ for both 01 and +1-1 encoding.
 \cref{app:lemma:n=3:01} proves that 01 encoding cannot express generalised XOR (the parity function) for $n>2$.
 \cref{app:lemma:n=3:+1-1} proves the analogous result for +1-1 encoding.
\end{proof}

\begin{lemma}[Number of functions a QNN can separate]\label{lemma:num_Boolean_functions_expressible}
 The number of functions on the amplitude encoded Boolean dataset of $n$ dimensions the QNN can express is upper bounded by 
 $$
 2^{n^3+n^2\log_2(e)+1},
 $$
 which as a fraction of the total $2^{2^n}$ tends to 0 as $n$ increases.
\end{lemma}

\begin{proof}
 The maximum possible number of functions a linear classifier in $K$ dimensions on $N$ datapoints can express is given by \cite{mackay2003information,baldi2019polynomial}
 \begin{equation}
  T(N,K)=
  \begin{cases} 
   2^N & \text{if } K \geq N \\
   2\sum^{K-1}_{k=0} {N-1\choose k}<2^{d^2\left( \log_2(e) + \log_2(p)-2\log_2(d)
 \right)} & \text{if } K < N
  \end{cases}
 \end{equation}
 For the QNN on the Boolean dataset, $N=2^n$ and $K=n^2$. This means
 \begin{equation}
  T = 2\sum^{n^2-1}_{k=0} {2^n-1\choose k}
 \end{equation}
 functions can be expressed. We can bound this partial sum with 
 $$
 T(2^n-1,n^2)<2\left[\frac{e(2^n-1)}{n^2}\right]^{n^2}< 2^{n^3+n^2\log_2(e)+1}.
 $$
 With $2^{2^n}$ functions overall, as $n$ increases the fraction of expressible functions will tend to near $0$.

 The perceptron satisfies $2^{n^2-n\log_2{n} -O(n)}< T(n)<2^{n^2-n\log_2{n} +O(n)}$ \citep{vsima2003general}.
\end{proof}

\section{Expressivity results for Multi-Layer QNNs}

\begin{lemma}[QNNs are no less expressive than Perceptrons when the data is in a single orthant]\label{app:lemma:QNN>Perceptron}
 Given some dataset $D$, if all $x\in D$ $\norm{x}=1$ and $x_i>0\; \forall \; x_i$ (i.e. the data has unit norm and is in the positive orthant), a QNN can express any function a perceptron can. However, if these conditions are relaxed, perceptrons can express functions that QNNs cannot.
\end{lemma}

\begin{proof}
 A perceptron expresses a function $f(x)=w\cdot x$. We can use a QNN to express any function a perceptron can express, assuming $\norm{W}_2=1$ and $\norm{x}_1=1$, and $x_i>0$ (i.e. all the data is in the positive orthant). Assume the perceptron has input dimension $2^n$ (corresponding to a QNN with $n$ input qubits). Assume we encode the data $x$ into the QNN using amplitude encoding on the square root of the inputs, i.e.\
 $$
 \ket{x} = (\sqrt{x_0}, \dots, \sqrt{x_{2^n}})
 $$
 Let $W=diag(w_0, \dots, w_{2^n})$
 \begin{equation}
 U=\frac{1}{\sqrt{2}}
 \begin{bmatrix}
 \sqrt{1+W} & -\sqrt{1-W} \\
 \sqrt{1-W} & \sqrt{1+W} \\
 \end{bmatrix}
 \end{equation}
 where $w=\text{diag}(w_i)$ (of dimension $2^n$), for a perceptron with input dimension $2^n$ and the square roots are performed elementwise. With a single readout qubit and measurement operator $Z = I^{2^n}\otimes Z_2$, where $Z_2=[[1, 0], [0, -1]]$, this expresses the function $f(x) = \bra{x} U^\dag Z U \ket{x}= w\cdot x$ provided the above conditions are satisfied.
 If $D$ contains $x$ and $-x$, then because for the QNN, $f(x)=f(-x)$ but for a perceptron $f(x)=-f(-x)$, we have a function expressible by the perceptron that is not expressible by the QNN.
\end{proof}

\cref{app:lemma:QNN>Perceptron} means the following.
Because perceptron expresses a function $f(x)=w\cdot x$, we can use a QNN to express any function a perceptron can express, assuming $|w_i|\leq 1$ and $x_i\in\{0, 1\}$.

\begin{equation}
U=\frac{1}{\sqrt{2}}
\begin{bmatrix}
\sqrt{1+w} & -\sqrt{1-w} \\
\sqrt{1-w} & \sqrt{1+w} \\
\end{bmatrix}
\end{equation}

where $w=\text{diag}(w_i)$ (of dimension $N$), for a perceptron with input dimension $N$ and the square roots are performed elementwise. 
If we are classifying $\{0,1\}^n$, we need $\ceil{\log_2{n}}$ qubits. The Hilbert space will therefore have dimension $d=2^{\ceil{\log_2{n}}}$ $1+w$ will be a $d\times d$ matrix. Using $w^i$ to denote a matrix satisfying $w_{ii}=1$ and $0$ otherwise, we need to consider $w^i$ $0<i\leq n$.
This means if we have an input $X=(x, 0)^T$, $UX=\frac{1}{\sqrt{2}}(\sqrt{1+w}x, \sqrt{1-w}x)^T$.

\begin{lemma}[Perceptrons can be simultaneously encoded on 2 output qubits]\label{lem:simul_perceptrons:2}
 For a QNN with two readout qubits, acting on input data $x\in \mathbb{R}^N$ in the positive orthant ($x_j>0$ for all $j$), the function $f(x) = [w^0\cdot x, w^1 cdot x]$ can be encoded by a QNN.
\end{lemma}

\begin{proof}
 \begin{equation}
  U_1 U_0=\frac{1}{\sqrt{2}}
  \begin{bmatrix}
  \sqrt{1+w^1} & 0 & -\sqrt{1-w^1} & 0\\
  0 & \sqrt{1+w^1} & 0 & -\sqrt{1-w^1} \\
  \sqrt{1-w^1} & 0 & \sqrt{1+w^1} & 0 \\
  0 & \sqrt{1-w^1} & 0 & \sqrt{1+w^1} \\
  \end{bmatrix}
  \frac{1}{\sqrt{2}}
  \begin{bmatrix}
  \sqrt{1+w^0} & -\sqrt{1-w^0} & 0 & 0\\
  \sqrt{1-w^0} & \sqrt{1+w^0} & 0 & 0 \\
  0 & 0 & \sqrt{1+w^0} & -\sqrt{1-w^0}\\
  0 & 0 & \sqrt{1-w^0} & \sqrt{1+w^0}\\
  \end{bmatrix}
 \end{equation}
 
 The action of these matrices on input $X=(x,0,0,0)^T$ produces 
 
 $$X'=(1/2)(\sqrt{1+w^1}\sqrt{1+w^0}x, \sqrt{1+w^1}\sqrt{1-w^0}x, \sqrt{1-w^1}\sqrt{1+w^0}x,\sqrt{1-w^1}\sqrt{1-w^0}x)^T$$
 
 Measurements on either qubit (with $Z_1 = Z \otimes I \otimes I^{d}$ or $Z_0 = I \otimes Z \otimes I^{d}$) on the joint system without affecting the distribution of either. The expectation value we want is $f_i(x) = X'^T Z_i X'$
 \begin{align}
  f_0(x)&=[(1+w^1)(1+w^0) - (1+w^1)(1-w^0) + (1-w^1)(1+w^0) - (1-w^1)(1-w^0)]x/4 \\
  &=(1+w_0)x-(1-w_0)x\\
  f_1(x)&=[(1+w^1)(1+w^0) + (1+w^1)(1-w^0) - (1-w^1)(1+w^0) - (1-w^1)(1-w^0)]x/4 \\
  &=(1+w_1)x-(1-w_1)x
 \end{align}
 meaning we can simultaneously encode two orthogonal outputs.
\end{proof}

\begin{theorem}[Readout qubits are independent]\label{app:thm:independent_readouts}
 A QNN with $d$ readout qubits receiving $x\in \mathbb{R}^N$, with $x_j>0$ for all $j$, can simultaneously encode $d$ perceptrons independently encoding $f(x)=w_i\cdot x$.
\end{theorem}

\begin{proof}
 We will prove this by induction, starting from $d=2$ proved in \cref{lem:simul_perceptrons:2}.

 Consider a QNN acting on $d$ output qubits with input $\ket{0}^d\ket{x}$. This means there exists some unitary matrix $U$ of dimension $N2^d $ such that 
 \begin{equation} %\label{eq:QNN(x)}
  f_i(x;\theta) = \mathbbm{1}(\bra{x}\bra{0}^d U^\dag(\theta) Z_i U(\theta)\ket{0}^d\ket{x}),
 \end{equation}
 where $Z_i=I_2 \otimes I_2 \dots \otimes Z \otimes \dots I_2 \otimes I_2$. 

 Define a new operator,
 $$
 U_{d+1} = \frac{1}{\sqrt{2}}\begin{bmatrix}
  \sqrt{1+w^{d+1}}\otimes I^{2^d} & -\sqrt{1-w^{d+1}}\otimes I^{2^d}\\
  \sqrt{1-w^{d+1}}\otimes I^{2^d} & \sqrt{1+w^{d+1}}\otimes I^{2^d}\\
  \end{bmatrix}
 $$
 
 Adding a further qubit with $U_{d+1}$ does not affect the function for qubits $1$ through $d$
 \begin{align}
 f_i(x)&=
  \begin{bmatrix}
   x^\dag & 0 \\
  \end{bmatrix}
  U_{d+1}^\dag
  \begin{bmatrix}
   U^\dag & 0 \\
   0 & U^\dag \\
  \end{bmatrix}
  \begin{bmatrix}
   Z_i & 0 \\
   0 & Z_i \\
  \end{bmatrix}
  \begin{bmatrix}
   U & 0 \\
   0 & U \\
  \end{bmatrix}
  U_{d+1}
  \begin{bmatrix}
   x \\
   0 \\
  \end{bmatrix}\\
  &= 1/2(x^\dag \sqrt{1+w^{d+1}} U^\dag Z_i U \sqrt{1+w^{d+1}} x + x^\dag \sqrt{1-w^{d+1}} U^\dag Z_i U \sqrt{1-w^{d+1}} x)\\
  &=x^\dag U^\dag Z_i U x
 \end{align}
 this follows due to the blocky nature of $x$.
 Then, measuring the final $(d+1)$'th qubit is straightforward:
 \begin{align}
  f_{d+1}(x)&=
  \begin{bmatrix}
   x^\dag & 0 \\
  \end{bmatrix}
  U_{d+1}^\dag
  \begin{bmatrix}
   U^\dag & 0 \\
   0 & U^\dag \\
  \end{bmatrix}
  \begin{bmatrix}
   I^{(N2^d)} & 0 \\
   0 & -I^{(N2^d)} \\
  \end{bmatrix}
  \begin{bmatrix}
   U & 0 \\
   0 & U \\
  \end{bmatrix}
  U_{d+1}
  \begin{bmatrix}
   x \\
   0 \\
  \end{bmatrix}=x^\dag U_{d+1}^\dag Z_{d+1} U_{d+1} x\\
  &= 1/2(x^\dag \sqrt{1+w^{d+1}} U^\dag Z_i U \sqrt{1+w^{d+1}} x - x^\dag \sqrt{1-w^{d+1}} U^\dag Z_i U \sqrt{1-w^{d+1}} x)\\
  &=w^{d+1} \cdot x
 \end{align}
 Therefore, by induction, the statement holds.
\end{proof}

\subsection{DQNN expressivity conditions}\label{app:sec:dqnn_proofs}

\cref{app:thm:independent_readouts} proves that a QNN with $n$ input qubits and dd output qubits can express any function of the form $f_i(x)=w_{ij}x_j$ with $n$ input and $d$ output neurons, provided the data is real, normalised and in the all-positive quadrant, and $|w_{ij}|\leq 1$. This is like the action of the first linear layer in a FCN. By applying a non-linear function to this output (and an optional bias term), a QNN with dd output qubits with data $x$ satisfying the above conditions can express the output of the first hidden layer of an FCN on $x$. Encoding this output for a second QNN with 1 readout qubit means we can express any function a 1 hidden layer FCN can (again, subject to data and weight norm constraints). This means the DQNNs proposed (\cref{def:DQNN}) would satisfy universal approximation theorems (e.g.\ \citep{hornik1989multilayer}).

\begin{definition}[DQNNs]\label{def:DQNN}
 A DQNN is analogous to a Fully Connected Neural Network in the classical world. The first layer is a standard QNN, but with $p$ readout qubits. The second layer is a standard QNN with $q$ input qubits (the exact value will depend on the encoding used in the second layer), and one readout qubit.
 The data would be encoded as $x^{(L=0)}=\ket{0}^{p}\ket{x}$. The output of this layer would be a vector $z$ of dimension $p$ with components $z_i$ the expectation values of the $i$'th readout qubit. The output of the first layer is $[\langle q_0 \rangle, \dots, \langle q_p \rangle]$. A trainable (classical) bias term $b_i$ can be added to each of these, followed by a non-linearity $\sigma$, to give $x^{(L=1)}=[\sigma(\langle q_0 \rangle+b_0), \dots, \sigma(\langle q_p \rangle+b_p)]$
 This output can then be encoded for the second layer with some encoding method $\phi^{(L=1)}(x^{(L=1)})$, and the second layer of the DQNN classifies this input like a regular QNN.
 In our experiments, the nonlinearity $\sigma$ is thresholding at 0, and we use two encodings $\phi^{(L=1)}$ in the second layer
 \begin{enumerate}
  \item Amplitude -- for $p$ readout qubits in the first layer, $x^{(L=1)}$ is encoded onto $q=\lceil \log_2 p\rceil$ qubits in the second layer using amplitude encoding. We denote this a \dqnna.
  \item Basis -- for $p$ readout qubits, $x^{(L=1)}$ is encoded onto $p$ qubits in the second layer using basis encoding. We denote a DQNN that uses this \dqnnb.
 \end{enumerate}
\end{definition}

\begin{lemma}[\dqnnb for Boolean classification are fully expressive]\label{lem:DQNN:Amp-Basis}
 Consider a QNN $Q_1$ with $n$ output qubits and $\ceil{log_2(n)}$ input qubits acting on the $n$-dimensional Boolean dataset with amplitude encoding. Take a second QNN $Q_2$ with $n$ input qubits and 1 output qubit. Encode $\ket{1}$ on the $i$'th qubit of $Q_2$ if $f_i(x)>b_i$ where $f_i(x)$ is the expectation value of the $i$'th readout qubit of $Q_1$, and $b_i\in [-1,1]$ is some threshold parameter, and encode $\ket{0}$ otherwise.
 This system $Q_{12}$ is capable of representing any function on $\{0,1\}^n$.
\end{lemma}

\begin{proof}
 By \cref{app:thm:independent_readouts} we can independently encode $n$ perceptrons with weights $[1,0,\dots]$, $[0,1,0,\dots]$, \dots, $[0,\dots,0,1]$. Let $b_i=0$. $Q_1$ therefore acts as a map from the amplitude encoding of $\{0,1\}^n$ to basis, so $Q_2$ will be fully expressive. Therefore, chaining two networks $Q_1$ and $Q_2$ together in this way can produce a QNN $Q_{12}$ that is fully expressive, even if $Q_1$ is not.
\end{proof}

\begin{lemma}[\dqnna for Boolean classification are fully expressive]\label{lem:DQNN:Amp-Amp}
 Consider a QNN $Q_1$ with $2^n$ output qubits and $\ceil{log_2(n)}$ input qubits acting on the $n$-dimensional Boolean dataset with amplitude encoding. Take a second QNN $Q_2$ with $n$ input qubits and 1 output qubit.
 Let $f_i(x)$ be the expectation value of the $i$'th readout qubit of $Q_1$ and $b_i\in [-1,1]$ some threshold parameter.
 Encode $\mathbbm{1}(f(x)>b)\in\mathbb{R^{2^n}}$ on the $n$ input qubits of $Q_2$ using amplitude encoding
 This system $Q_{12}$ is capable of representing any function on $\{0,1\}^n$.
\end{lemma}

\begin{proof}
 By \cref{app:thm:independent_readouts} we can independently encode $2^n$ perceptrons -- equivalent to a fully connected layer of width $2^n$. The final layer can represent a perceptron of input width $2^n$.
 Lemma 5.1 in \citep{mingard2019neural} shows that an FCN acting on unnormalised input data is fully expressive. Normalising the inputs requires us to fix the bias terms $b_i=1$.
\end{proof}

\FloatBarrier

\section{QNN data}

In this appendix, we show data complementary to \cref{fig:qnn_posterior,fig:dqnn_posterior}. As QNNs are kernel methods, their posterior is almost exclusively determined by their prior and likelihood (with training details playing a more limited role, unlike in the case of DNNs). While the posterior contains most of the relevant information about their inductive biases, studying the prior still can be instructive (see e.g. \citet{valle2018deep,mingard2023deep}).

\cref{fig:qnn:prior:haar,fig:dqnn:prior} show the priors of QNNs with different encodings and the two types of DQNN, respectively. In each figure, (a) shows the $n=5$ Boolean system, and (b) the $n=7$ system. For $n=5$ there are $2^{32}\sim 10^9$ functions, meaning an appreciable fraction of those can be found in $10^8$ samples. In contrast, there are $2^{128}\sim 10^{40}$ functions for the $n=7$ system, so the vast majority of these will never be found. In each of the four subfigures, we produce the following plots. 

\textbf{The first panel} shows the prior probability $P(f)$ of a function $f$ plotted against their rank, $R(f)$. The prior $P(f)$ is the probability of randomly initialising to $f$, where parameters are sampled from sensible initialisation distributions (e.g.\ for QNNs, the Unitary matrix is sampled from the Haar measure). The most probable function has rank 1, the second most rank 2, and so on. 
When the relation $P(f)\propto R(f)^{-\alpha}$ has $\alpha \neq 1$, \citet{ridout2024bounds} argue that a learning agent cannot learn as the number of training datapoints tends to infinity. \cref{fig:qnn:prior:haar} indicates that all QNNs (and the perceptron) will fail to generalise for some large $n$ (assuming $\alpha$ is constant as $n$ increases). \cref{fig:dqnn:prior} shows that the FCN does satisfy $P(f)\propto Rank(f)^{-1}$. \dqnna has the largest value of $\alpha$, but it is still less than 1.

\textbf{The second panel} shows a scatter plot of $P(f)$ v.s.\ the Lempel-Ziv complexity of $f$ (see \citep{mingard2023deep} for details). In \cref{fig:qnn:prior:haar}, we can see that the perceptron and QNN on the amplitude encoding are much more strongly biased towards simple functions than the QNN acting on the ZZ, RT and basis encodings. \cref{fig:dqnn:prior} shows that the \dqnna has a similar, but slightly weaker bias than the FCN, with the \dqnnb being substantially weaker.

\textbf{The third panel} shows $P(LZ(f))$, the probability that the model initialises to a function with some LZ complexity v.s.\ that complexity $LZ(f)$ (i.e.\ created by coarse-graining functions in the second panel by complexity). It is another way of showing which encodings and learning agents have the strongest simplicity bias.

\textbf{The fourth and final panel} shows the $LZ(f)$ v.s. $T(f)$ (where $T(f)$ is equal to min(number 1s, number 0s) in $f$) of the functions found during sampling. This is best interpreted for the $n=7$ datasets, where we can see that the FCN manages to find functions with low LZ complexity but high $T$ (simple class-balanced functions). The perceptron, amplitude-encoded QNN and \dqnna all find some of those functions. However, all other encodings only find simple class imbalanced functions, or complex functions with high class balance. This type of simplicity bias towards functions with low class balance is relatively trivial.

\begin{figure}[t]
 \centering
 \begin{subfigure}[b]{1.0\textwidth}
    \centering
	\includegraphics[width=\columnwidth]{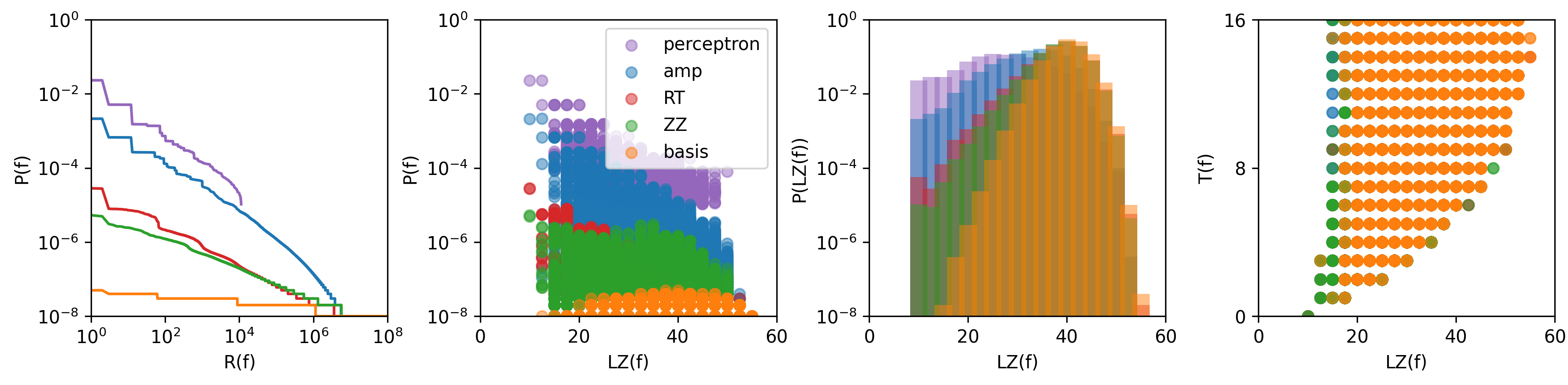}
  \caption{n=5}
  \label{fig:qnn:prior:haar:n=5}
 \end{subfigure}
 \begin{subfigure}[b]{1.0\textwidth}
    \centering
	\includegraphics[width=\columnwidth]{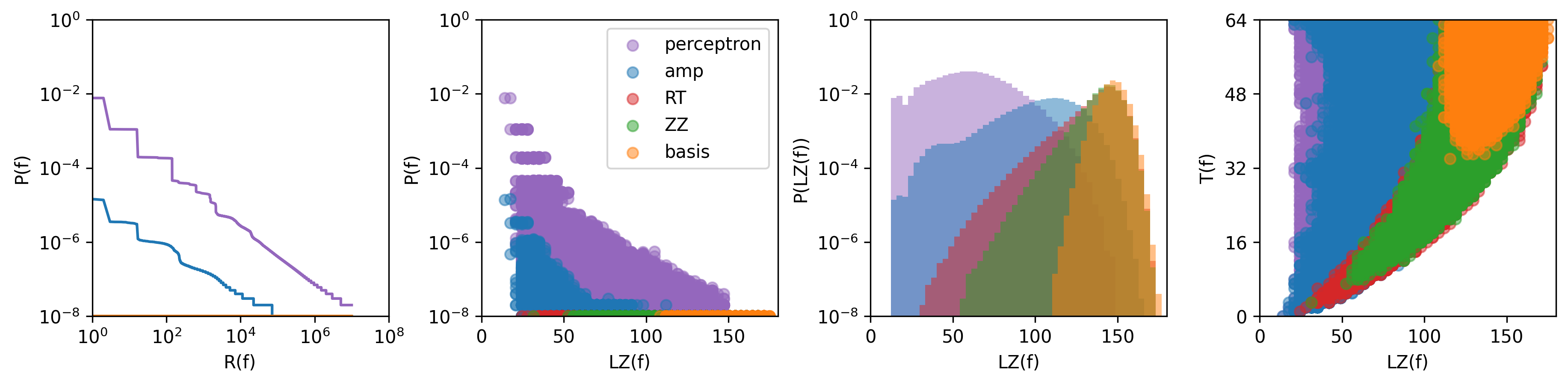}
  \caption{n=7}
  \label{fig:qnn:prior:haar:n=7}
 \end{subfigure}
 \caption{\textbf{QNN priors for the $\mathbf{n=5}$ and $\mathbf{n=7}$ Boolean dataset.} The encodings are the same as described in \cref{fig:qnn_posterior}, with the RT and ZZ encodings using $n$ input qubits. $T(f)$ denotes the minimum number of 0s or 1s in the function (unnormalised class balance).
 The prior probability $P(f)$ was approximated by sampling a unitary matrix of the appropriate dimension ($2\times 2^{\lceil \log(n)\rceil}$ for amplitude encoding, otherwise $2^{2n}$) from the Haar measure, and evaluating the function expressed by the QNN on all $2^n$ inputs. $10^8$ samples were taken for each learning system.
 The perceptron and QNN on amplitude encoding have the most similar inductive biases, with the perceptron still being noticeably stronger. This discrepancy largely vanishes when a more restrictive set of gates is used (see \cref{fig:qnn:prior:n=5:actual}). The QNN on the RT and ZZ encodings are significantly less biased towards simple functions than the QNN on amplitude encoding. The first plots in each row show $P(f)$ v.s.\ rank for functions ordered by the value of $P(f)$.
 Note that $P(f)$ for basis encoding is not constant due to finite-size effects (but is straightforwardly constant). See \cref{fig:dqnn:prior} for the FCN, \dqnna and \dqnnb priors. }\label{fig:qnn:prior:haar}
\end{figure}

\begin{figure}[t]
 \centering
 \begin{subfigure}[b]{1.0\textwidth}
 \centering
	\includegraphics[width=\columnwidth]{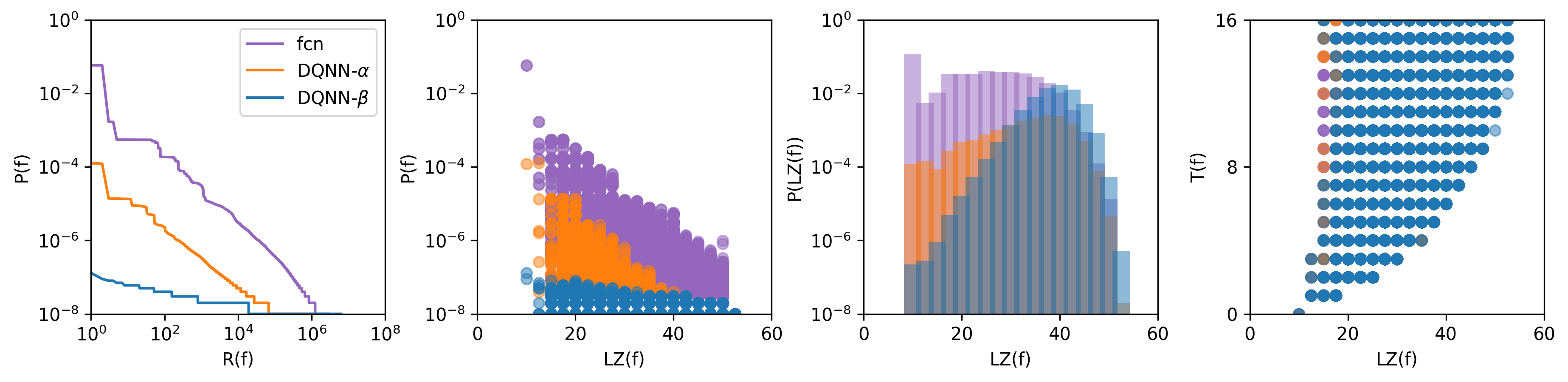}
  \caption{n=5}
  \label{fig:dqnn:prior:n=5}
 \end{subfigure}
 \begin{subfigure}[b]{1.0\textwidth}
 \centering
	\includegraphics[width=\columnwidth]{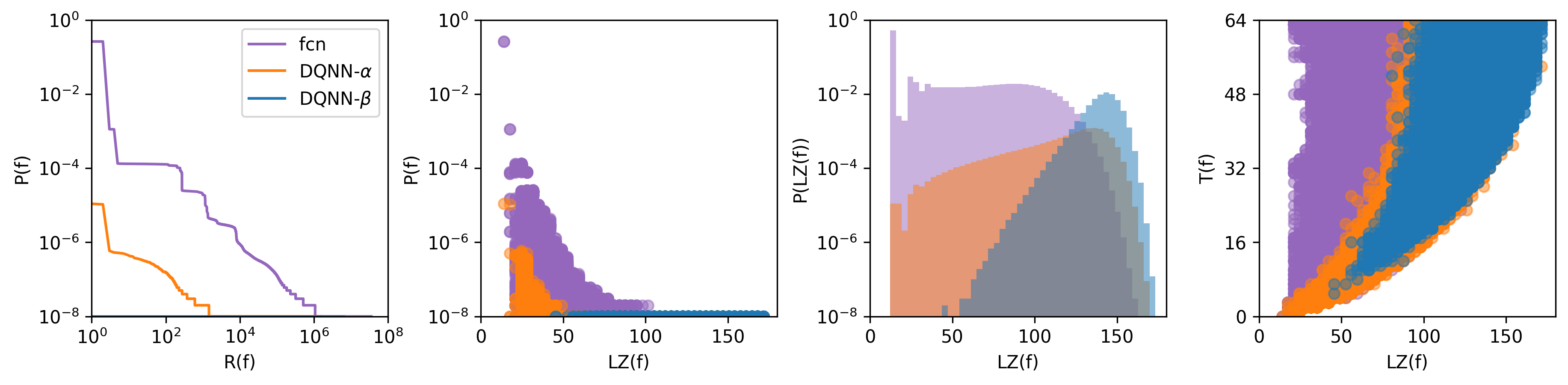}
  \caption{n=7}
  \label{fig:dqnn:prior:n=7}
 \end{subfigure}
 \caption{\textbf{FCN and DQNN priors for the $\mathbf{n=5}$ and $\mathbf{n=7}$ Boolean dataset.} The setup of the DQNNs is exactly as described in \cref{fig:dqnn_posterior}. 
 The prior probability $P(f)$ was approximated by sampling a unitary matrix for each layer of the DQNN from the Haar measure, and evaluating the function expressed by the DQNN on all $2^n$ inputs (for the FCN, weights and biases were randomly initialised using PyTorch defaults). $10^8$ samples were taken for each learning system.
 $T(f)$ denotes the minimum number of 0s or 1s in the function (unnormalised class balance).
 The first plots in each row show $P(f)$ v.s.\ rank for functions ordered by the value of $P(f)$.
 We compare the 2 types of DQNNs to a 1-layer FCN with ReLU activations, and see that the FCN is more biased towards simple functions than the \dqnna, which is in turn more biased than the \dqnnb. The third panel in each subfigure highlights how significant the difference in inductive bias is between the FCN and \dqnna. It is possible that a \dqnna with a restricted set of gates would perform similarly to the FCN (as observed for the QNN and perceptron in \cref{fig:qnn:prior:n=5:actual}).
 }\label{fig:dqnn:prior}
\end{figure}

\FloatBarrier

\section{Expressivity and inductive bias of ansatz expressing restricted unitaries}\label{sec:charlie}

While the main body of this paper has focused on the expressivity and inductive bias of QNNs modelled as arbitrary unitaries, in practise quantum machine learning researchers generally use parameterised ansatz that can only express a much more restricted class of unitaries. While these ansatz have lower expressivity than the class of arbitrary unitaries, they may also have improved inductive bias. In this section, we investigate the inductive bias of several common ansatz over the Boolean dataset. 

In \cref{fig:qnn:prior:n=5:actual} we plot the prior for the $n=5$ Boolean system for a QNN which can express any Boolean function, but not any arbitrary unitary. It takes the gate structure from \citep{ngoc2020tunable} but uses controlled-U gates in place of controlled-NOT gates. We used Qiskit to instantiate the network.
Qualitatively, its prior is similar to the prior of the QNNs that can express arbitrary unitaries (\cref{fig:qnn:prior:haar}). 

\cref{fig:qnn:posterior:n=5:actual} shows the test error for $n=5$ Boolean data for the same architecture (see \cref{fig:qnn_posterior} the $n=7$ Boolean system using the TPP correspondence). We use the same basic 100 target functions from \cref{fig:qnn_posterior}, but cut down to length 32 (as the $n=5$ system has 32 inputs rather than 128).
Many of the functions are missing for each encoding method, as they did not converge to 0 training error in the maximum number of epochs we used (200). 

Next, we look at the five most expressive and five least expressive ansatz from Figure 3. in \citet{Sim_2019}, to determine how the ansatz's inductive bias is affected by the expressivity. We use Boolean inputs of dimension $5$, encoded with both amplitude and basis encoding that uses $n=3$ and $n=5$ qubits respectively. By simulating the ansatz with TorchQuantum \cite{hanruiwang2022quantumnas} and randomly sampling the parameters from $[-\pi, \pi]$ uniformly, we can characterise the probability of expressing different functions on the Boolean dataset.

\begin{figure}[t]
 \centering
 \includegraphics[width=\columnwidth]{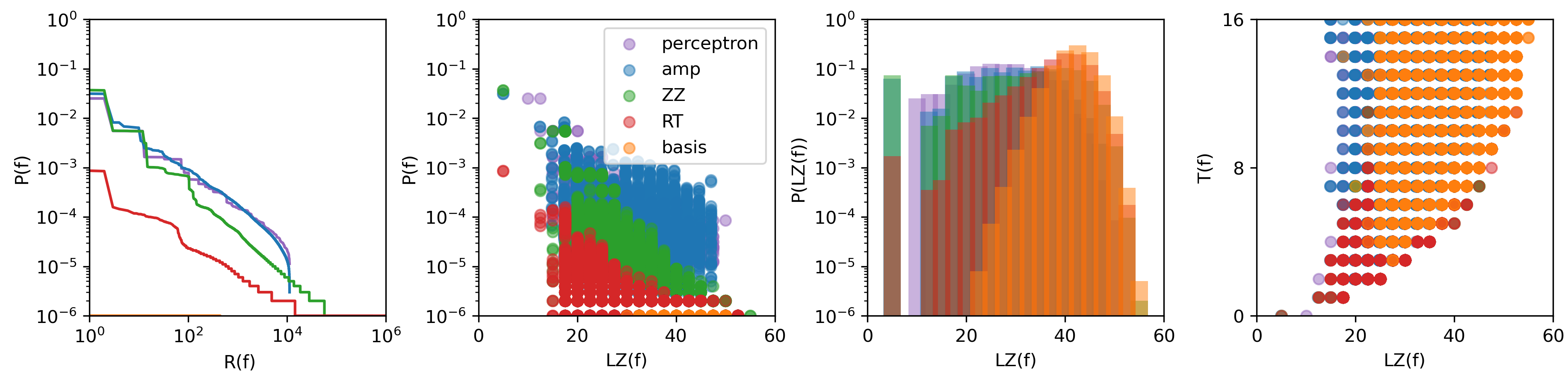}
 \caption{\textbf{QNN priors for $\mathbf{n=5}$ with restricted gates.} This experiment shows data for an implementation of a QNN (using Qiskit) with the gate structure described in \citep{ngoc2020tunable}, but replacing their controlled-not gates with controlled-U. Parameters were sampled uniformly from $[-\pi,\pi]$. $T(f)$ denotes the minimum number of 0s or 1s in the function (unnormalised class balance). $10^8$ samples were taken for each model to approximate $P(f)$. 
 See \cref{fig:qnn:prior:haar} for a description of $P(f)$ and associated quantities.  The first plot in the  row shows $P(f)$ v.s.\ rank for functions ordered by the value of $P(f)$.
 The plots are qualitatively similar to those in \cref{fig:qnn:prior:haar:n=5}, but there are a number of important differences. The first panel shows steeper power laws for each encoding, indicating a stronger bias towards simple functions. The difference is particularly noticeable for the ZZ encoding. $P(f)$ for the QNN on amplitude encoding and the perceptron appears almost identical. Predictably, basis encoding shows no inductive bias.
 }\label{fig:qnn:prior:n=5:actual}
\end{figure}

\begin{figure}[t]
 \centering
\includegraphics[width=\columnwidth]{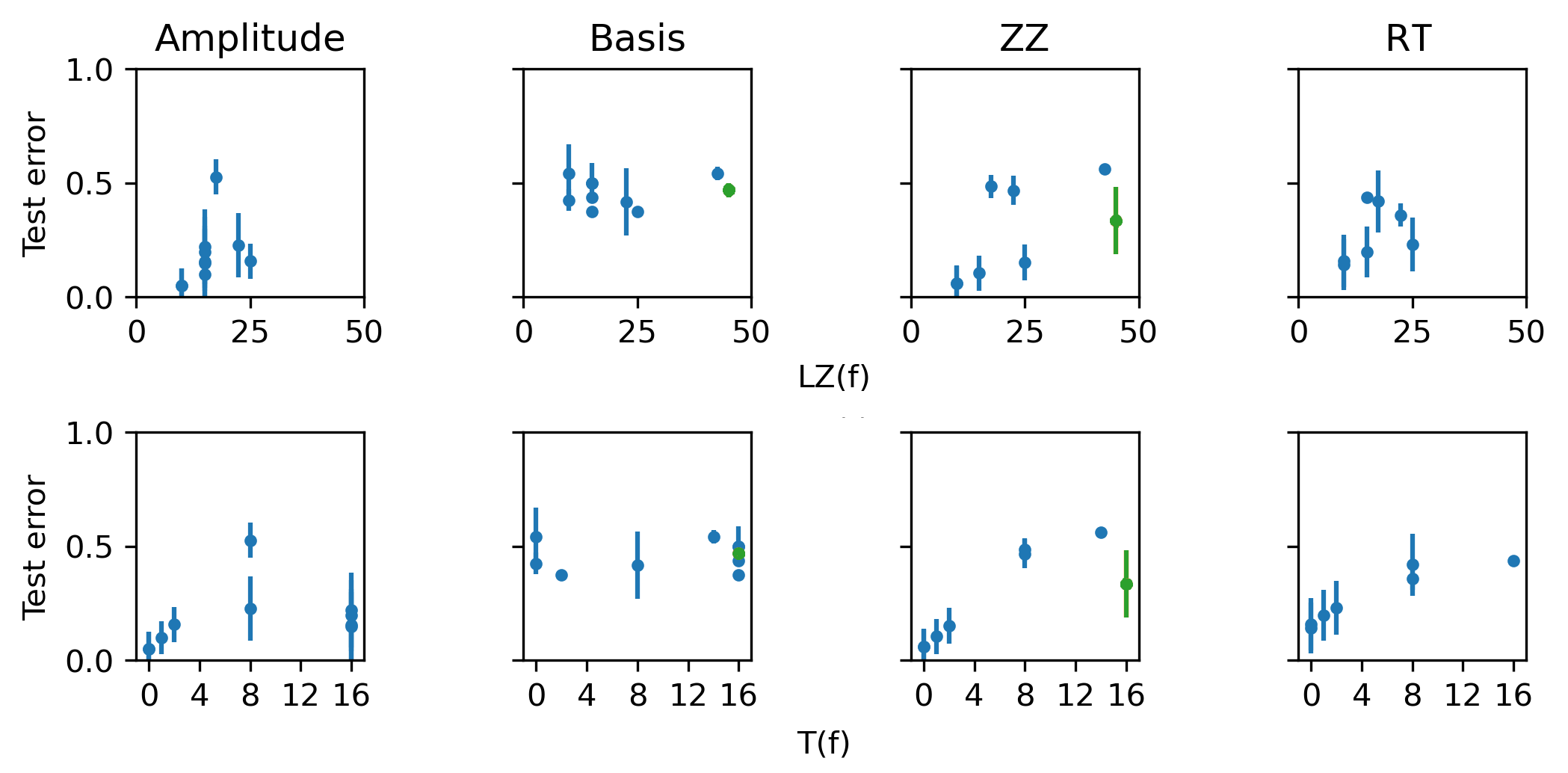}
 \caption{\textbf{Trained QNNs using Qiskit for $\mathbf{n=5}$.} We train on a subset of the functions used in \cref{fig:qnn_posterior} (cut off at 32 since $n=5$). We use a quantum circuit from \citep{ngoc2020tunable}, but with controlled-U gates in place of their controlled-NOT gates. $T(f)$ denotes the minimum number of 0s or 1s in the function, whichever is lowest (unnormalised class balance). 
 Only functions whose training converged to zero training error were kept, which turned out to be 
  a small subset of the total number of functions. QNNs suffer from optimisation problems such as barren plateaus. This means that one cannot easily distinguish between cases where a QNN cannot reach zero error because it fundamentally can't express a function, and cases where it cannot reach zero error due to optimisation problems. This issue is one reason why a strategy such as mapping to the TPP, which is much easier to train,  is needed in order to reliably probe questions relating to generalisation  and expressivity in QNNs.   Although the issues around QNN optimisation make it difficult to draw firm conclusions,  we do observe some parallels between the  performance of a Qiskit trained QNN on these functions and the TPP on the $n=7$ Boolean dataset in \cref{fig:qnn_posterior}. For example, no generalisation occurs for basis encoding and the ZZ and RT encodings appear biased towards functions with low class balance.  The ZZ encoding performs comparatively well on parity function, but does not reach zero error as it does with the TPP for $n=7$.   The QNN with amplitude encoding shows what may potentially be a more interesting inductive bias, for example,  it generalises relatively well on $f=1010\dots$ which has high class-balance, but low LZ complexity.  Nevertheless,  optimisation problems mean that one cannot draw reliable conclusion based on this data alone.
 }
 \label{fig:qnn:posterior:n=5:actual}
\end{figure}

\begin{figure}[t]
 \centering
 \begin{subfigure}[b]{1.0\textwidth}
	\includegraphics[width=\columnwidth]{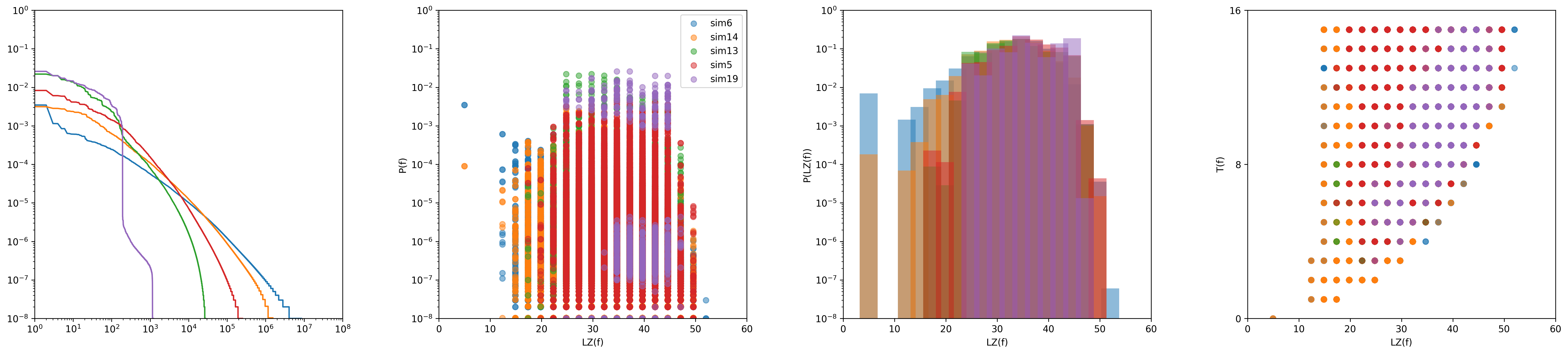}
  \caption{Amplitude encoding, most expressive}
  \label{fig:qnn:ansatz:amp:most}
 \end{subfigure}
 \begin{subfigure}[b]{1.0\textwidth}
	\includegraphics[width=\columnwidth]{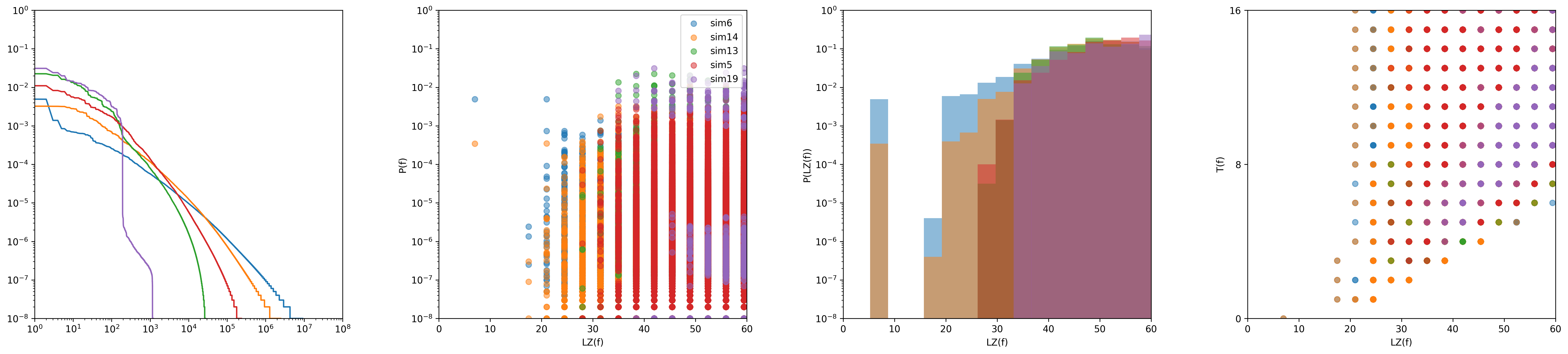}
  \caption{Basis encoding, most expressive}
  \label{fig:qnn:ansatz:basis:most}
 \end{subfigure}
 \begin{subfigure}[b]{1.0\textwidth}
 \includegraphics[width=\columnwidth]{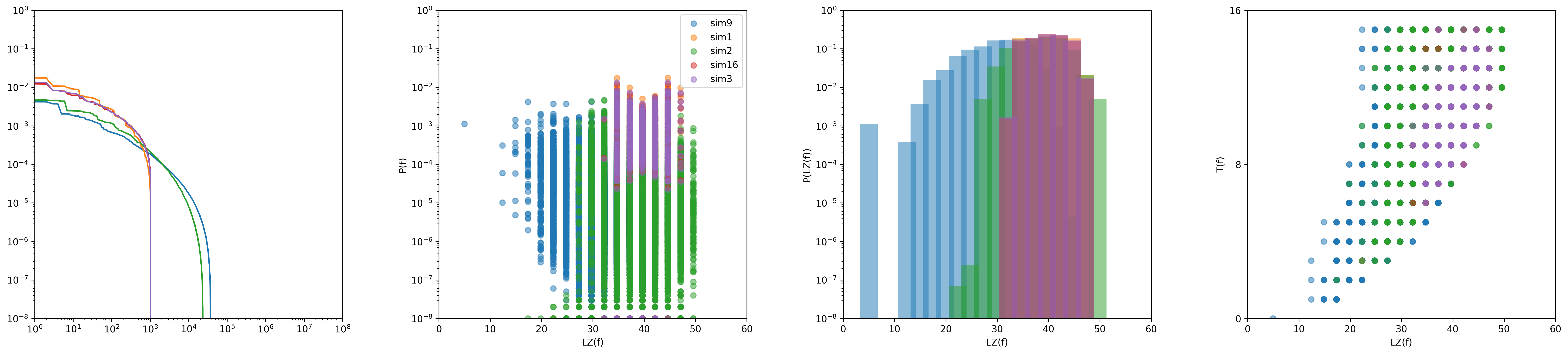}
  \caption{Amplitude encoding, least expressive}
  \label{fig:qnn:ansatz:amp:least}
 \end{subfigure}
 \begin{subfigure}[b]{1.0\textwidth}
 \includegraphics[width=\columnwidth]{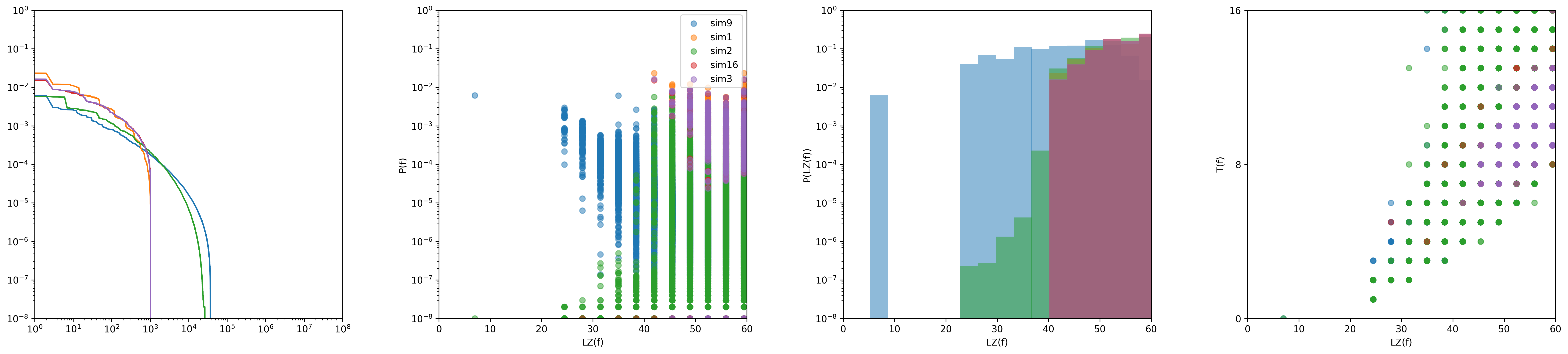}
  \caption{Basis encoding, least expressive}
  \label{fig:qnn:ansatz:basis:least}
 \end{subfigure}
 \caption{\textbf{QNN priors for the $\mathbf{n=5}$ Boolean dataset with restricted ansatze.} The priors $P(f)$ are calculated by randomly sampling rotation parameters from $[-\pi, \pi]$ for QNNs with gate structures from \citet{Sim_2019}. 
 % more detail
 $T(f)$ denotes the minimum number of 0s or 1s in the function (unnormalised class balance).}\label{fig:qnn:ansatz}
\end{figure}

We can see in \cref{fig:qnn:ansatz:amp:least} and \cref{fig:qnn:ansatz:basis:least} that the least expressive ansatzes have a strong bias towards a small subset of functions, but cannot express the vast majority of the $\sim 10^9$ possible functions on the $n=5$ Boolean dataset, even when using basis encoding. Any functions past the steep decline in probability in these figures are not expressible by the restricted ansatze. The bias is towards simple functions, but basis encoding misses some of these simple functions, such as 1-local parity ($f=010101\dots$).

With amplitude encoding, the most expressive ansatze (see \cref{fig:qnn:ansatz:amp:most}) have a similar inductive bias to the space of all unitaries, as previously seen in \cref{fig:qnn:prior:haar:n=5}. This is a relatively strong inductive bias towards simple functions, but these ansatze will suffer from more limited expressivity than the space of all unitaries.

Using the most expressive ansatze may provide some advantage in the case of basis encoding, as it induces a strong inductive bias not present when using the space of all unitaries (see \cref{fig:qnn:ansatz:basis:most} and \cref{fig:qnn:prior:haar:n=5}). However, all inputs to this model are still orthogonal in Hilbert space, and applying the model to an input amounts to selecting one of the columns of the $2^n \times 2^n$ unitary matrix. Any correlations in the output will entirely reflect correlations in the model parameters. In order to ensure that these correlations faithfully reflect the correlations in the input data, the basis encoding must be done so that the encoding of similar points corresponds to correlated columns in the weight matrix, which requires human intervention.

In all of these cases, the number of parameters of the ansatze scale linearly with the number of qubits, and as the input space becomes larger the expressivity over the exponential Hilbert space must necessarily decrease.

\FloatBarrier

\begin{figure}[H]
  \centering
  \begin{subfigure}[t]{0.24\textwidth}
   \centering
   \includegraphics[width=\textwidth]{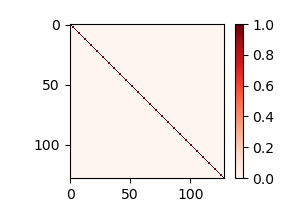}
   \caption{Basis}
   \label{fig:kernel_e0}
  \end{subfigure}
  \begin{subfigure}[t]{0.24\textwidth}
   \centering
   \includegraphics[width=\textwidth]{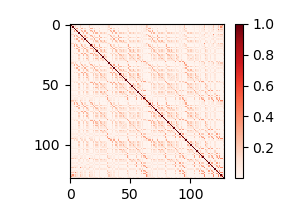}
   \caption{ZZ ($2^n$) Feature Map}
   \label{fig:kernel_e1}
  \end{subfigure}
  \begin{subfigure}[t]{0.24\textwidth}
   \centering
   \includegraphics[width=\textwidth]{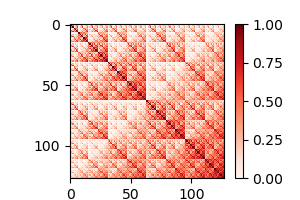}
   \caption{Amplitude}
   \label{fig:kernel_e4}
  \end{subfigure}
  \begin{subfigure}[t]{0.24\textwidth}
   \centering
   \includegraphics[width=\textwidth]{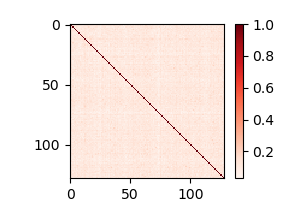}
   \caption{RT ($2^n$)}
   \label{fig:kernel_e3}
  \end{subfigure}

  \begin{subfigure}[t]{0.24\textwidth}
   \centering
   \includegraphics[width=\textwidth]{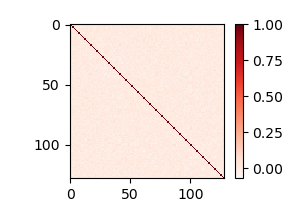}
   \caption{Basis (Empirical)}
   \label{fig:emp:kernel_e0}
  \end{subfigure}
  \begin{subfigure}[t]{0.24\textwidth}
   \centering
   \includegraphics[width=\textwidth]{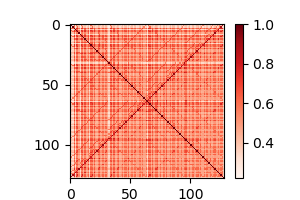}
   \caption{ZZ ($2^n$) Feature Map (Empirical)}
   \label{fig:emp:kernel_e1}
  \end{subfigure}
  \begin{subfigure}[t]{0.24\textwidth}
   \centering
   \includegraphics[width=\textwidth]{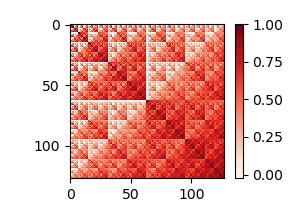}
   \caption{Amplitude (Empirical)}
   \label{fig:emp:kernel_e3}
  \end{subfigure}
  \begin{subfigure}[t]{0.24\textwidth}
   \centering
   \includegraphics[width=\textwidth]{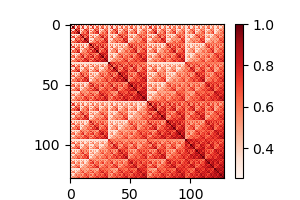}
   \caption{FCN}
   \label{fig:emp:kernel_e4}
  \end{subfigure}
  \caption{ 
  \textbf{QNN kernel visualisations.} with different encodings of the Boolean system compared to the FCN. The top row contains analytic kernels when $U(\theta)$ is sampled from the Haar measure. The bottom row shows the empirical kernel $K(x_i,x_j)=\mathbb{E}_{\theta}\left[ f(x_i)f(x_j)\right]$ from QNNs implemented with Qiskit (that use the gate structure from \citep{ngoc2020tunable} but with Controlled-U gates in place of controlled-NOT gates) which are unable to express every possible $U(\theta)$. The basis and amplitude encoding look very similar in both the theoretical and empirical cases, but the empirical ZZ kernel is much more strongly biased than the analytic one.}\label{app:subsec:kernel_visualise}
\end{figure}

\section{Quantum Kernels}\label{sec:quantum_kernels}

A kernel function $k$ maps n-dimensional inputs $x_i, x_j$ into a higher dimensional space $k(x_i, x_j) = \langle f(x_i), f(x_j) \rangle$, where $f$ is a map from $n$-dimension to $m$ dimension space. A kernel matrix $K$ can be constructed from all the input data: $K_{ij} = k(x_i,x_j)$.
A quantum feature map $\phi(x)$ maps classical data, represented as a vector $x$, to a quantum Hilbert space and the Kernel matrix for a QNN is $K_{ij} = \left| \langle \phi^\dagger(x_j)| \phi(x_i) \rangle \right|^{2}$.
See \cref{app:subsec:kernel_visualise} for visualisations of the kernels for QNNs on the Boolean dataset.

As shown in \cref{sec:introduction}, QNNs are linear models with a specific embedding $\phi$ (feature map). Kernel methods can be used to solve linear models with non-linear embeddings, which have been widely studied in the context of Support vector machines (SVM)s \citep{scholkopf2002learning} and more recently in infinite-width DNNs \citep{jacot2018neural, canatar2021spectral}. In this section, we briefly review kernel methods and rephrase the arguments regarding the expressivity and inductive bias in the main text in the language of kernels. 

\subsection{Kernel methods}
A kernel method can be described as a linear classifier with a feature map (embedding), a map from the raw inputs (e.g. abstract Boolean dataset) to typically higher dimensional embedding space (e.g. Hilbert Space). In short, solving kernel regression with $K(x^1,x^2) = \Phi^T(x^1)\Phi(x^2)$ is equivalent to solving a linear regression on $x^{(i)}=\phi({x'}^{(i)})$. Thus, any solution of kernel regression can be expressed as 
\begin{equation}\label{eq:linear_model}
f(x;w) = \sum_{k=1}^p w_k\Phi_k(x'),
\end{equation}
where $w \in \mathbb{R}^p$ are the parameters. The kernel $K(x^1,x^2)$ is the inner product between $\Phi(x^1)$ and $\Phi(x^2)$ in the embedding space, equivalently the space spanned by $[\Phi_1, \Phi_2, \cdots, \Phi_p]$. Note that $p$ can be taken to infinity. Kernels have the advantage that non-linear tasks can be solved linearly (and often analytically), but performance relies on the choice of the feature map (i.e. embedding). For more details, see \citep{kung2014kernel}.

As seen in \cref{eq:linear_model}, only a linear combination of $[\Phi_1, \Phi_2, \cdots, \Phi_p]$ can be expressed with this model. However, a measure on the function domain is required for a set of functions to form a vector space, which in our case would be the true distribution of the raw input data (i.e. $q$ such that $x' \sim q$). With $q$, we define $\mathcal{H}$, a function space expressed by the model, and an integral operator $T: \mathcal{H} \rightarrow \mathcal{H}$ given as
\begin{equation}\label{eq:new3}
T[f](x^{(2)}) := \int_\chi \Phi(x^{\prime (1)})^T\Phi(x^{\prime(2)})f(x^1)q(x^1)dx^1 = \int_\chi K(x^{(1)},x^{(2)})f(x^{(1)})q(x^{(1)})dx^{(1)}.
\end{equation}
By Mercer's theorem \citep{mercer1909xvi}, the integral operator $T$ can be decomposed into $p$ eigenfunctions $e_k: \chi \rightarrow \mathbb{R}$ and eigenvalues $\rho_k \ge 0$ such that 
\begin{equation}\label{eq:new}
T[e_k] = \rho_ke_k,\quad\quad k \in [1, \cdots, p],
\end{equation}
or in more familiar bra-ket notation as 
\begin{equation}\label{eq:new2}
T = \sum_{k=1}^p\rho_k \ket{e_k}\bra{e_k}.
\end{equation}
This basis in $e_k$ is a more natural one to analyse the expressivity and inductive bias of a kernel than the original basis $\Phi_k$.

\subsection{Expressivity of kernel methods}
The dimension of the function space $\mathcal{H}$ is the dimension of the span of $[\Phi_1, \Phi_2, \cdots, \Phi_p]$ or the number of non-zero eigenvalues of $T$, and is a measure of expressivity of the kernel.
In \cref{fig:kernel_qnn_percept}, we plot the normalized eigenvalue spectra of different kernels (equivalently different embeddings) for the $n=7$ Boolean dataset. A perceptron on the Boolean system has a rank of only $7$ and a QNN kernel on amplitude embedding has a rank of $28$, showing the limited size of the function space compared to other embeddings (e.g.\ ZZ or Basis), which has the full rank of $128$. 
Note that the maximum number of eigenvalues for a QNN kernel with $N$-dimensional Hilbert space is $N^2$. For amplitude encoding with $n=7$, it is already clear that no more than 49 eigenfunctions can have non-zero eigenvalues (and in practice, there are only 28 of them).
The experiment suggests that some Boolean functions cannot be expressed both by the perceptron and QNN with amplitude encoding, as shown in \cref{sec:qb:exp_bool}.

\subsection{Strength of inductive bias of kernels}
The generalisation loss measures the MSE loss of the raw outputs to the targets (note that this is different from generalisation error -- the fraction of misclassified examples in the test set) of kernel regression is known analytically \citep{canatar2021spectral, jacot2020kernel, cui2021generalization, harzli2021double}, and the inductive bias of kernels can be quantified by the eigenvalues and eigenfunctions \citep{canatar2022kernel, simon2021theory}. Kernels learn eigenfunctions with larger eigenvalues faster. This means that a faster decay rate of the eigenvalue spectra indicates a stronger inductive bias (towards the eigenfunctions with the large eigenvalues), while flat eigenvalue spectra indicate a uniform inductive bias over functions. 

The eigenvalue spectra have been used in prior studies to argue about the inductive bias and how many data points are required for learning. In \cite{kubler2021inductive, huang2021power}, the lack of inductive bias due to uniform eigenvalues was pointed out and an exponential number of data points, compared to the number of qubits, are required for learning. In \cite{canatar2022bandwidth} (a study of classically produced kernels) and \citep{holmes_connecting_2022} (for quantum produced kernels), to compensate for the lack of inductive bias, the eigenvalue spectra were `unflattened' and allowed learning with only polynomial data points.

In \cref{fig:kernel_qnn_percept}, the decay rate of the eigenvalues for different embeddings (thus kernels) are plotted. The kernel with basis encoding has uniform eigenvalue spectra (indicating no inductive bias) while the FCN embedding (of the effective FCN kernel) has the fastest decay of eigenvalues (strongest inductive bias). The consequences of the inductive bias are illustrated in \cref{fig:kernel_generalisation}(a), where we plot the learning curve for a uniform Boolean dataset with $n=7$ trained with different embeddings on the target function $f=010101\dots$. Note that for the basis encoding kernel with a uniform inductive bias, the generalisation loss improves negligibly with additional data points.

\subsection{Alignment between the kernel and the target function}

Strong inductive bias indicates the model is biased toward certain functions over others, but for a strong inductive bias to imply better generalisation, the model's bias must be towards the target function. In the study of kernels, one way of measuring how well  the model is biased towards the target function  by using the  task-model alignment measure \citep{canatar2021spectral} \begin{equation}
TA(k) = \frac{1}{\bra{f^*}\ket{f^*}} \sum_i^k \bra{e_i}\ket{f^*}^2,
\end{equation}
where $1\le k \le p$ and $e_i$ are the eigenvalues of the kernel in descending order of the eigenvalues.
The effect of alignment is demonstrated in \cref{fig:kernel_generalisation}(a) and \cref{fig:kernel_generalisation}(b) where the target functions are $0101\dots$ and parity respectively. The FCN kernel performs much better in (a), as the task-model alignment is greater than (b). In addition, note that in (b), the FCN embedding performs worse than the basis embedding, which has a uniform inductive bias.   More generally, for kernel methods, engineers try different kernels, or tweak hyperparameters to obtain better alignment with the target function \citep{rasmussen2003gaussian}. The process is feature engineering, and in QNNs it is equivalent to finding different embeddings or introducing different hyperparameters such as bandwidth.
\begin{figure}[H] 
 \centering
 \includegraphics[width=0.5 \textwidth]{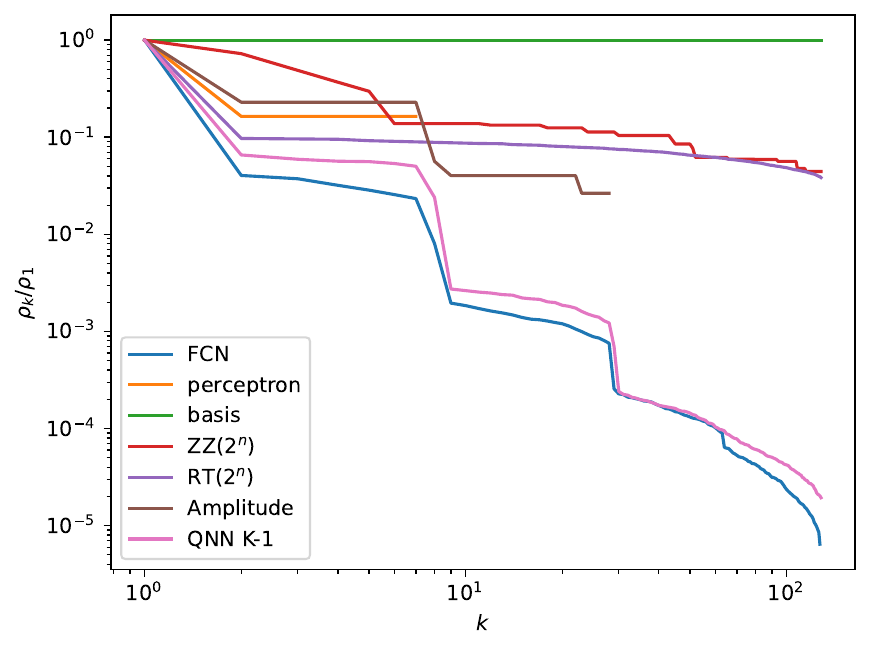}%

 \caption{\textbf{Eigenvalues of $T$ (kernels) with different encodings on $n=7$ Boolean dataset.} For each given encoding, a linear operator $T$ and its eigenvalues can be calculated via \cref{eq:linear_model}. perceptron (no encoding) and amplitude encoding have ranks of 7 and 28 respectively, while other encodings have a rank of 128. The eigenvalues decay the fastest for FCN encoding and slowest (uniform) for basis encoding.
 }\label{fig:kernel_qnn_percept}
  \bigskip
\end{figure}

\begin{figure}[H] 
 \centering
  \subfloat[simple function (f=$0101\dots$)]{%
   \includegraphics[width=0.49 \textwidth]{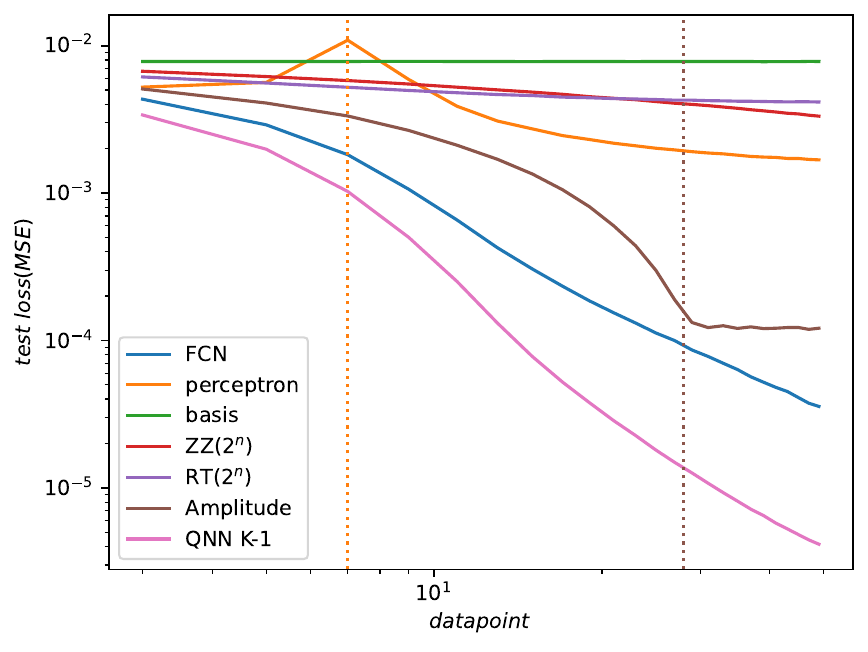}%
  }
  \subfloat[parity function]{%
   \includegraphics[width=0.49 \textwidth]{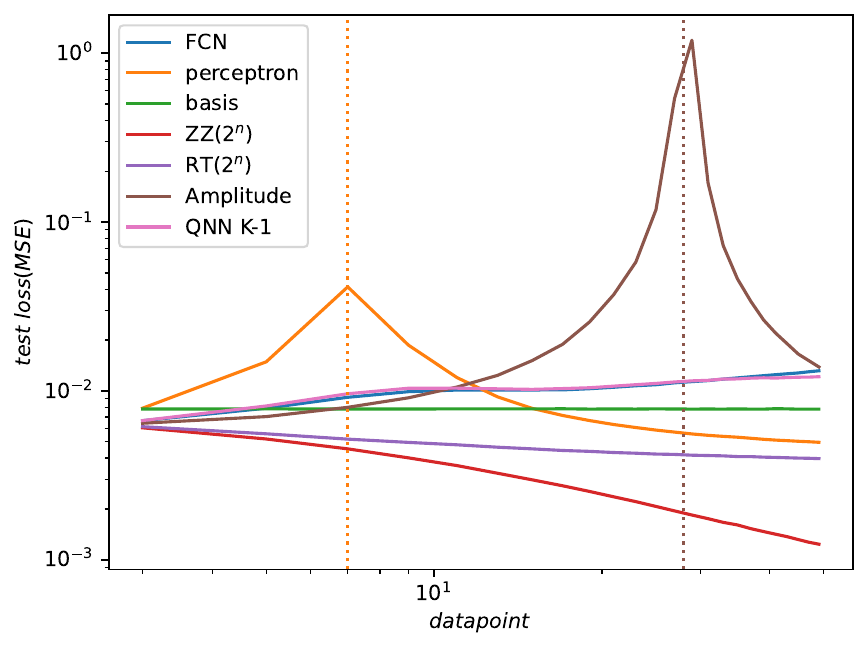}%
  }
 \caption{\textbf{Learning curve for different kernels on a $n=7$ Boolean dataset with different target functions.} Kernel ridgeless regression was performed using a varying number of training data points. The mean square error (MSE) generalisation loss was calculated using 50 unseen data points over 5000 independently and identically distributed (i.i.d.) sampled training sets. The dotted vertical lines are the maximum rank for respective kernels. (a) The target function $f=0101\dots$. FCN and amplitude encoding have relatively good good task-model alignment compared to other encodings, which is reflected in their better performance. Basis encoding, with its uniform inductive bias, does improve performance with additional datapoints. (b) Test loss if the target function is the parity function. The task-model alignment for the FCN encoding is poor for this function, leading to a larger loss compared to basis encoding which has a uniform inductive bias. The parity function is not within the hypothesis set expressible by the perceptron and amplitude encoding kernels, leading to double descent behaviour for both encodings near $n=rank(T)$. 
 }\label{fig:kernel_generalisation}
  \bigskip
\end{figure}

\FloatBarrier